\documentclass{article} 
\usepackage{macros}
\usepackage{authblk}

\title{Actor-Critic Algorithms for Learning Nash Equilibria in $\bm{N}$-player General-Sum Games}

\author[1]{Prasad H.L. \thanks{prasad@astrome.co}}
\author[2]{Prashanth L A \thanks{prashla@umd.edu}}
\author[3]{Shalabh Bhatnagar \thanks{shalabh@csa.iisc.ernet.in}}
\affil[1]{\small Astrome Technologies Pvt Ltd,
  Bangalore, INDIA}
\affil[2]{\small Institute for Systems Research, 
  University of Maryland, College Park, US}
\affil[3]{\small Department of Computer Science and Automation,
Indian Institute of Science, Bangalore, INDIA}
\date{}

\begin{document}

\maketitle

\begin{abstract}
We consider the problem of finding stationary Nash equilibria (NE) in a finite discounted general-sum stochastic game. 
We first generalize a non-linear optimization problem from \cite{filar-vrieze} to a $N$-player setting and break down this problem into simpler sub-problems that ensure there is no Bellman error for a given state and an agent. 
We then provide a characterization of solution points of these sub-problems that correspond to Nash equilibria of the underlying game and for this purpose, we derive a set of necessary and sufficient SG-SP (Stochastic Game - Sub-Problem) conditions. Using these conditions,
we develop two \textit{actor-critic} algorithms: OFF-SGSP (model-based) and ON-SGSP (model-free). 
Both algorithms use a \textit{critic} that estimates the value function for a fixed policy and an \textit{actor} that performs descent in the policy space using a descent direction that avoids local minima.  We establish that both algorithms converge, in self-play, to the equilibria of a certain ordinary differential equation (ODE), whose stable limit points coincide with stationary NE  of the underlying general-sum stochastic game. On a single state non-generic game (see \cite{hart2005stochastic}) as well as on a synthetic two-player game setup with $810,000$ states, we establish that ON-SGSP consistently outperforms NashQ \citep{huwellman2003} and FFQ \citep{littman} algorithms.   
\end{abstract}

\keywords
{General-sum discounted stochastic games, Nash equilibrium, multi-agent reinforcement learning, multi-timescale stochastic approximation.
}

\section{Introduction}
Traditional game theoretic developments have been for single-shot games where all agents participate and perform their preferred actions, receive rewards and the game is over. However, several emerging applications have the concept of multiple stages of action or most often the concept of {\em time} in them. One intermediate class of games to handle multiple stages of decision is called {\em repeated games}. However, repeated games do not provide for characterizing the influence of decisions made in one stage to future stages. Markov chains or Markov processes are a popular and widely applicable class of random processes which are used for modeling practical systems. Markov chains allow system designers to model {\em states} of a given system and then model the time behavior of the system by {\em connecting} those states via suitable probabilities for transition from {\em current state} to a {\em next state}. A popular extension to Markov chains which is used for modeling optimal control scenarios is the 
class of Markov Decision Processes ({\em MDPs}). Here, in a given state, an {\em action} is allowed to be selected from a set of available actions. Based on the choice of the action, a suitable {\em reward/cost} is obtained/incurred. Also, the action selected influences the probabilities of transition from one state to another. \cite{shapley} merged these concepts of MDPs (or basically Markov behavior) and games to come up with a new class of games called as {\em stochastic games}. In a stochastic game, all participating agents select their actions, each of which influence the rewards received by all agents as well as the transition probabilities. Since the inception of stochastic games by \cite{shapley}, they have been an important class of models for multi-agent systems. A comprehensive treatment of stochastic games under various pay-off criteria is given by \cite{filar-vrieze}. Many problems like fishery games, advertisement games and several oligopolistic situations can be modelled as stochastic games~\
citep{
breton1986computation,filar-vrieze,olley1992dynamics,pakes1998empirical,pakes2001stochastic,bergemann1996learning}. 

We consider a finite {\em stochastic game} (also referred to as Markov game (cf. \cite{littman}) setting that evolves over discrete time instants. As illustrated in Fig. \ref{fig:rl-model}, at each stage and in any given state $x \in \S$, all agents act simultaneously with an action vector $a \in \A(x)$ resulting in a transition to the next state $y \in \S$ according to the transition probability $p(y|x,a)$ as well as a reward vector $\r(x,a)$. No agent gets to know what the other agents' actions are before selecting its own action and the reward $r^i(x,a)$ obtained by any agent $i$ in each stage depends on both system state $x$ (common to all agents) and the aggregate action $a$ (which includes other agents' actions). Each individual agent's {\em sole} objective is maximization of his/her {\em value function}, i.e., the expected discounted sum of rewards. However, the transition dynamics as well as the rewards depend on the actions of 
all agents and hence, the dynamics of the game is coupled and not independent. We assume that $\r(x, a)$ and the action vector $a$ are made known to all agents after every agent $i$ has picked his/her action $a^i$ - this is the canonical {\em model-free} setting\footnote{While the ON-SGSP algorithm that we propose is for this setting, we also propose another algorithm - OFF-SGSP - that is model based.}. However, we do not assume that each agent knows the other agents' policies, i.e., the distribution from which the actions are picked.

\begin{figure}
\centering
\tikzstyle{block} = [draw, rectangle,  line width=0.2mm,join=round,minimum height=1.5cm, minimum width=5cm]
\tikzstyle{smallblock} = [draw, rectangle, minimum height=2em, minimum width=2em]
\scalebox{0.9}{
\begin{tikzpicture}
\node [smallblock,fill=green!20, rectangle] at (0,0) (node1) {$\bm{1}$};
\node [smallblock,fill=green!20, rectangle, right=1cm of node1] (node2) {$\bm{2}$};
\node [right=1cm of node2] (dots) {$\dots$};
\node [smallblock,fill=green!20, rectangle, right=1cm of dots] (nodeN) {$\bm{N}$};
\node at (2.5,-1) (on-sgsp) {{\large\bf Agents}};
\node [block, fill=red!20] at (2.5,4) (env)  {{\large\bf Environment}};

\draw [color=gray,thick](-1.1,-1.3) rectangle (6.1,1.3);
\begin{pgfonlayer}{background}
    \filldraw [line width=4mm,join=round,brown!20]
(-1,-1.2) rectangle  (6,1.2);
  \end{pgfonlayer}

\draw[>=latex',->,thick] (3.2,1.3) to [out=45,in=-45] node [right] {\textbf{Action} $\mathbf{a=\left < a^1, a^2, \dots, a^N \right >}$} (env);
\draw[>=latex',->,thick] (env) to [out=-135,in=135] node [left] {\makecell{\textbf{Reward }$\mathbf{r=\left < r^1, r^2, \dots, r^N \right >}$,\\ \textbf{next state }$\bm{y}$}} (1.7,1.3);

\end{tikzpicture}}
\caption{Multi-agent RL setting}
\label{fig:rl-model}
\end{figure}
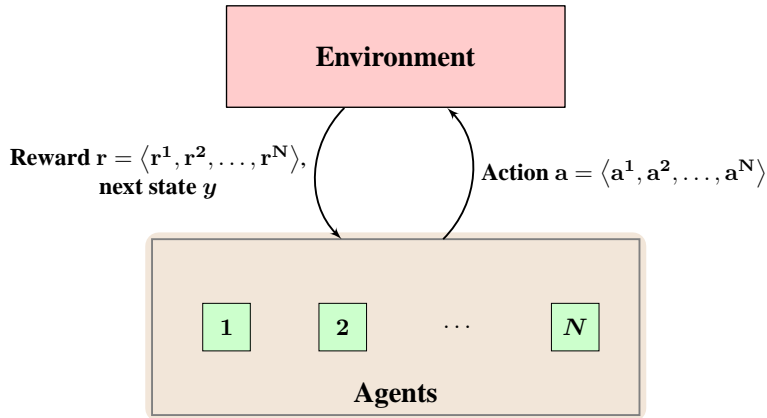

The central concept of stability in a stochastic game is that of a {\em Nash equilibrium}. 
At a Nash equilibrium point (with a corresponding Nash strategy), each agent plays a {\em best-response} strategy assuming all the other agents play their equilibrium strategies (see Definition \ref{def:nash} for a precise statement). 
This notion of equilibrium makes perfect sense in a game setting where agents do not have any incentive to unilaterally deviate from their Nash strategies. 


\cite{breton1991algorithms,filar-vrieze} establish that finding the stationary NE of a two-player discounted stochastic game is equivalent to solving an optimization problem with a non-linear objective function and linear constraints. 
We extend this formulation to general $N$-player stochastic games and observe that this generalization causes the constraints to be non-linear as well.   
Previous approaches to solving the optimization problem have not been able to guarantee convergence to global minima, even for the case of $N=2$. 
 In this light, our contribution is significant as we develop an algorithm to find a global minimum for any $N\ge 2$ via the following steps:\footnote{A preliminary version of this paper, without the proofs, was published in \textit{AAMAS 2015} - see \cite{prasad2015two}. In comparison to the conference version, this paper includes a more detailed problem formulation, formal proofs of convergence of the two proposed algorithms, some additional experiments and a revised presentation.} 
\begin{description}
\item[\textbf{Step 1 (Sub-problems):}] We break down the main optimization problem into several sub-problems. Each sub-problem can be seen as ensuring no Bellman error, for a particular state $x \in \S$ and agent $i \in \{1,\ldots,N\}$, where $\S$ is the state space and $N$ is the number of agents of the stochastic game considered.
\item[\textbf{Step 2 (Solution points):}] We provide a characterization of solution points that correspond to Nash equilibria of the underlying game. As a result, we also derive a set of necessary and sufficient conditions, henceforth referred to as SG-SP (Stochastic Game - Sub-Problem) conditions.
\item[\textbf{Step 3 (Descent direction):}] Using SG-SP conditions, we derive a descent direction that avoids local minima. This is not a steepest descent direction, but a carefully chosen descent direction specific to this optimization problem, which ensures convergence only to points of global minima that correspond to SG-SP points (and hence Nash strategies). 
\item[\textbf{Step 4 (Actor-critic algorithms):}] We propose algorithms that incorporate the aforementioned descent direction to ensure convergence to stationary NE of the underlying game. 
\end{description}

The algorithms that we propose are as follows:
\begin{description}
\item[OFF-SGSP.] This is an offline, centralized and model-based scheme, i.e., it assumes that the transition structure of the underlying game is known.
\item[ON-SGSP.] This is an online, model-free scheme that is decentralized, i.e., learning is localized to each agent with one instance of ON-SGSP running on each agent. ON-SGSP only requires that other agents' actions and rewards are observed and not their policies, i.e., maps from states to actions.
\end{description}
We make the assumption that for all strategies, the resulting Markov chain is irreducible and positive recurrent. 
This assumption is common to the analysis of previous multi-agent RL algorithms as well (cf. \cite{huwellman,littman})\footnote{ For the case of stochastic games where
there are multiple communicating classes of states or even transient states, a possible work-around is to re-start the game periodically in a random state.}.
To the best of our knowledge, ON-SGSP is the first model-free online algorithm that converges in self-play to stationary NE for any finite discounted general-sum stochastic game where the aforementioned assumption holds.

As suggested by \cite{bowling2001rational}, two desirable properties of any multi-agent learning algorithm are as follows:
\begin{enumerate}[(a)] 
\item {\em Rationality\footnote{The term {\em rationality} is not to be confused with its common interpretation in economics parlance.}:} Learn to play optimally when other agents follow stationary strategies; and
 \item {\em Self-play convergence:} Converge to a Nash equilibrium assuming all agents are using the same learning algorithm.
 \end{enumerate}
  
Our ON-SGSP algorithm can be seen to meet both the properties mentioned above.
However, unlike the repeated game setting of \citep{bowling2001rational,conitzer2007awesome}, ON-SGSP solves  discounted general-sum stochastic games and possesses theoretical convergence guarantees as well.

 \tikzset{
    >=stealth',
    punkt/.style={
           rectangle,
           rounded corners,
           draw=black, very thick,
           text width=10em,
           minimum height=3.5em,
           text centered},
    pil/.style={
           ->,
           thick,
           shorten <=2pt,
           shorten >=2pt,}
}
\begin{figure}
 \centering
\scalebox{0.85}{\begin{tikzpicture}[node distance=1cm, auto,]
 \node (dummy) {};
 \node[punkt,fill=blue!20,above=1cm of dummy] (peval) {\bf Critic \\ (Policy Evaluation)};
 \node[punkt, fill=red!20,inner sep=5pt,below=1.5cm of dummy]
 (pimp) {\bf Actor\\ (Policy Improvement)};
 \node[right=3cm of dummy] (t) {\bf Value $\bm{v^{\pi^i}}$}
   edge[pil,<-,bend right=45] (peval.east) 
   edge[pil, bend left=45] (pimp.east); 
 \node[left=3cm of dummy] (g) {\bf Policy $\bm{\pi^i}$}
   edge[pil, bend left=45] (peval.west)
   edge[pil,<-, bend right=45] (pimp.west);
\end{tikzpicture}}
\caption{Operational flow of our algorithms}
\label{fig:flow}
\end{figure}
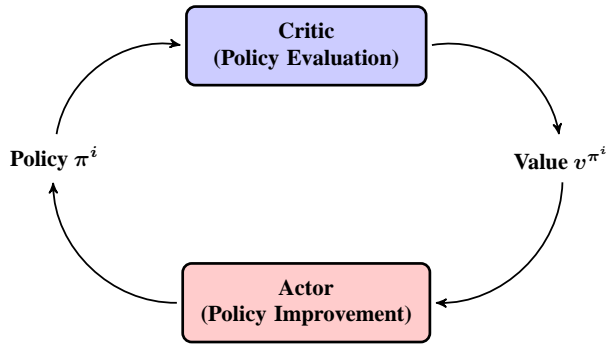

As illustrated in Fig. \ref{fig:flow}, the basic idea in both OFF-SGSP and ON-SGSP is to employ two recursions, referred to as the actor and the critic, respectively. Conceptually, these can be seen as two nested loops that operate as follows:
\begin{description}
\item[Critic recursion:] This performs policy evaluation, i.e., estimates the value function for a fixed policy. In the model-based setting (i.e., of OFF-SGSP), this corresponds to the well-known dynamic programming procedure - value iteration. On the other hand, in the model-free setting (i.e., of ON-SGSP), the critic is based  on temporal difference (TD) learning \citep{sutton}. 
\item[Actor recursion:] This incrementally updates the policy using gradient descent. For this purpose, it uses a descent direction that ensures convergence to a global minimum (and hence NE) of the optimization problem we mentioned earlier.
\end{description}
Using multi-timescale stochastic approximation (see Chapter 6 of \cite{borkar2008stochastic}) both the recursions above are run simultaneously, albeit with varying step-size parameters and this mimics the nested two-loop procedure outlined above.

The formal proof of convergence requires considerable sophistication, as we base our approach on the ordinary differential equations (ODE) method for stochastic approximation \citep{borkar2008stochastic}. 
While a few previous papers in the literature have adopted this approach (cf. \cite{akchurina},\cite{weibull}), their results do not usually start with an algorithm that is shown to track an ODE. Instead, an ODE is reached first via analysis and an approximate method is used to solve this ODE. On the other hand, we adopt the former approach and show that both OFF-SGSP and ON-SGSP converge using the following steps:
 \begin{enumerate}[\bfseries 1.]
\item Using two-timescale stochastic approximation, we show that the value and policy updates on the fast and slow timescales, converge respectively to the limit points of a system of ODEs.
\item Next, we provide a simplified representation for the limiting set of the policy ODE and use this to establish that the asymptotically stable limit points of the policy ODE correspond to SG-SP (and hence Nash) points. 
\end{enumerate}
While the first step above uses a well-known result (Kushner-Clark lemma) for analysing stochastic approximation recursions, the techniques used in the second step above are quite different from those used previously. The latter step is crucial in establishing overall convergence, as the strategy $\pi$ corresponding to each stable limit gives a stationary NE of the underlying general-sum discounted stochastic game.   
 
We demonstrate the practicality of our algorithms on two synthetic two-player setups. The first is a single state non-generic game adopted from \cite{hart2005stochastic} that contains two NEs (one pure, the other mixed), while the second is a \textit{stick-together game} with $810,000$ states (to the best of our knowledge, previous works on general-sum games have never considered state spaces of this size). On the first setup, we show that ON-SGSP always converges to NE, while NashQ \citep{huwellman2003} and FFQ \citep{littman} do not in a significant number of experimental runs. On the second setup, we show that ON-SGSP outperforms NashQ and FFQ, while exhibiting a relatively quick convergence rate - requiring approximately $21$ iterations per state. 

\paragraph{Map of the paper.}
 Section \ref{sec:lit-review} reviews relevant previous works in the literature. Section \ref{sect:prelims} formalizes a general-sum stochastic game and sets the notation used throughout the paper. Section \ref{sect:problem} formulates a general non-linear optimization problem for solving stochastic games and also form sub-problems whose solutions correspond to Nash strategies. Section \ref{sec:sgsp} presents necessary and sufficient SG-SP conditions for Nash equilibria of general-sum stochastic games. Section \ref{sect:off-sgsp} presents the offline algorithm OFF-SGSP, while Section \ref{sect:on-sgsp} provides its online counterpart ON-SGSP. Section \ref{sect:convergence} sketches the proof of convergence for both the algorithms. Simulation results for the single state non-generic game and the stick-together game settings are presented in Section \ref{sect:simulation}. Finally, concluding remarks are provided in Section \ref{sect:conclusions}.

\section{Related Work}
\label{sec:lit-review}
Various approaches have been proposed in the literature for computing Nash equilibrium of general-sum discounted stochastic games and we discuss some of them below.

\textbf{Multi-agent RL.}
Littman \citep{littman1994markov} proposed a
minimax Q-learning algorithm for two-player zero-sum stochastic games. Hu and Wellman \citep{huwellman,huwellman2003} extended the Q-learning approach to general-sum games, but their algorithms do not provide meaningful convergence guarantees.  Friend-or-foe Q-learning (FFQ) \citep{littman} is a further improvement based on Q-learning and with guaranteed convergence. However, FFQ converges to Nash equilibria only in restricted settings (See conditions A and B in \citep{littman}).
Moreover, the approaches in \citep{huwellman,huwellman2003} require computation of Nash equilibria of a bimatrix game, while the approach of  \cite{littman} requires solving a linear program, in each round of their algorithms and this is a computationally expensive operation. In contrast, ON-SGSP does not require any such equilibria computations. \cite{zinkevich2006cyclic} show that the traditional Q-learning based approaches are not sufficient to compute Nash equilibria in general-sum games\footnote{We avoid this impossibility result by searching for both values and policies instead of just values, in our proposed algorithms.}. 

\textbf{Policy hill climbing.}
This is a category of previous works that is closely related to ON-SGSP algorithm that we propose. Important references here include \cite{bowling2001rational}, \cite{bowling2005convergence}, \cite{conitzer2007awesome} as well as \cite{zhang2010multi}. All these algorithms are gradient-based, model-free and are proven to converge to NE for stationary opponents in self-play. However, these convergence guarantees are for repeated games only, i.e., the setting is a single state stochastic game, where the objective is to learn the Nash strategy for a stage-game (see Definition $1$ in \citep{conitzer2007awesome}) that is repeatedly played. 
On the other hand, we consider general-sum stochastic games where the objective is to learn the best-response strategy against stationary opponents in order to maximize the value function (which is an infinite horizon discounted sum). Further, we work with a more general state space that is not restricted to be a singleton. 

\textbf{Homotopy.}
\cite{herings2004stationary} propose an algorithm where a homotopic path between equilibrium points of $N$ independent MDPs and the $N$ player stochastic game in question, is traced numerically. This, in turn, gives a Nash equilibrium point of the stochastic game of interest. This approach is extended by \cite{herings2006homotopy} and \cite{borkovsky2010user}.
OFF-SGSP shares similarities with the aforementioned homotopic algorithms in the sense that both are
\begin{enumerate}[\bfseries(i)]
\item offline and model-based as they assume complete information (esp. the transition dynamics) about the game; and 
\item the computational complexity for each iteration of both algorithms grows exponentially with the number of agents $N$. 
\item Further, both algorithms are proven to converge to stationary NE, though their approach adopted is vastly different. OFF-SGSP is a gradient descent algorithm designed to converge to the global minimum of a nonlinear program, while the algorithm by \cite{herings2004stationary} involves a tracing procedure to find an equilibrium point.
\end{enumerate}

\textbf{Linear programming.}
\cite{mac2009solving} solve stochastic games by formulating intermediate optimization problems, called Multi-Objective Linear Programs (MOLPs). However, the solution concept there is of {\em correlated equilibria} and Nash points are a strict subset of this class (and hence are harder to find). 
Also, the complexity of their algorithm scales exponentially with the problem size. 

Both homotopy and linear programming methods proposed by \cite{mac2009solving} and \cite{herings2004stationary} are tractable only for small-sized problems. 
The computational complexity of these algorithms may render them infeasible on large state games. In contrast, ON-SGSP is a model-free algorithm with a per-iteration complexity that is linear in $N$, allowing for practical implementations on large state game settings (see Section \ref{sect:simulation} for one such example with a state space cardinality of $810,000$). 
We however mention that per-iteration complexity alone is not sufficient to quantify the performance of an algorithm - see Remark \ref{remark:complexity}. 

\textbf{Rational learning.} 
A popular algorithm with guaranteed convergence to Nash equilibria in general-sum stochastic games is rational learning, proposed by \cite{kalai1993rational}. 
In their algorithm, each agent $i$ maintains a prior on what he believes to be other agents' strategy and updates it in a Bayesian manner. Combining this with certain assumptions of absolute continuity and grain of truth, the algorithm there is shown to converge to NE. 
ON-SGSP operates in a similar setting as that in \citep{kalai1993rational}, except that we do not assume the knowledge of reward functions. ON-SGSP is a model-free online algorithm and unlike \citep{kalai1993rational}, any agent's strategy in ON-SGSP does not depend upon Bayesian estimates of other agents' strategies and hence, their absolute continuity/grain of truth assumptions do not apply. 

\textbf{Evolutionary algorithm.}
\cite{akchurina} employs numerical methods in order to solve a system of ODEs and only establishes empirical convergence to NE for a group of randomly generated games.
In contrast, ON-SGSP is a model-free algorithm that is provably convergent to NE in self-play. 
We also note that the system of ODEs given by \cite{akchurina} (also found in \citep[pp. 189]{weibull}) turns out to be similar to a portion of the ODEs that are tracked by ON-SGSP. 



\begin{remark}
\cite{shoham2003multi} and \cite{shoham2007if} question if Nash equilibrium is a useful solution concept for general-sum games. However, if we are willing to concede that {\em prescriptive, equilibrium agenda} is indeed useful for stochastic games, then we believe our work is theoretically significant. Our ON-SGSP algorithm is a prescriptive, co-operative learning algorithm that observes a sample path  from the underlying game and converges to stationary NE. To the best of our knowledge, this is the first algorithm to do so, with proven convergence.
\end{remark}


\section{Formal Definitions}
\label{sect:prelims}
A stochastic game can be seen to be an extension of the single-agent Markov decision process. A discounted reward stochastic game is described by a tuple $<N, \S , \mathcal{A}, p, \r, \beta>$, where $N$ represents the number of agents, $\S$ denotes the state space and $\A = \cup_{x \in \S} \A(x)$ is the aggregate action space, where $\A(x)= \prod\limits_{i =1}^N \mathcal{A}^i(x)$ is the Cartesian product of action spaces $(\A^i(x))$ of individual agents when the state of the game is $x \in \S$. We assume both state and action spaces to be finite.
Let $p(y|x,a)$ denote the probability of going from the current state $x \in \S$ to $y \in \S$ when the vector of actions $a \in \mathcal{A}(x)$ (of the $N$ players) is chosen and let $\r(x,a) = \left < r^i(x, a) : i = 1, 2,\dots,N \right >$ denote the vector of reward functions of all agents when the state is $x \in \S$ and the vector of actions $a \in \mathcal{A}(x)$ is chosen.
Also, $0 < \beta < 1$ denotes the discount factor that controls the influence of the rewards obtained in the future on the agents' strategy (see Definition \ref{def:vf} below).

\paragraph{Notation.}
$\left < \cdots \right >$ represents a column vector and $\underline{1}_{m}$ is a vector of ones with $m$ elements. 
The various constituents of the stochastic game considered are denoted as follows:\footnote{We use the terms {\em policy} and {\em strategy} interchangeably in the paper.}\\
\begin{description}
\item[\textbf{Action:}] $a = \left < a^1, a^2, \dots, a^N \right > \in \A(x)$ is the aggregate action, $a^{-i} $ is the tuple of actions of all agents except $i$ and
$\A^{-i}(x) := \prod\limits_{j \ne i} \A^j(x)$ is the set of feasible actions in state $x \in \S$ of all agents except $i$.\\
\item[\textbf{Policy:}] $\pi^i(x, a^i)$ is the probability of picking action $a^i \in \A^i(x)$ by agent $i$ in state $x \in \S$,\\
$\pi^i(x) = \left < \pi^i(x,a^i) : a^i \in \mathcal{A}^i(x) \right >$ is the randomized policy vector in state $x \in \S$ for the agent $i$,
$\pi^i = \left < \pi^i(x) : x \in \S \right >$,
$\pi = \left < \pi^i : i = 1, 2, \dots, N \right >$ is the strategy-tuple of all agents and\\
$\pi^{-i} = \left < \pi^j : j = 1, 2, \dots, N, j \ne i \right >$ is the strategy-tuple of all agents except agent $i$. We focus only on {\em stationary strategies} in this paper, as suggested by Theorem \ref{thm:nash-exist}.  \\
\item[\textbf{Transition Probability:}]
Let $\pi(x,a) = \prod\limits_{i = 1}^{N} \pi^i(x, a^i)$ and $\pi^{-i}(x,a^{-i})=\prod\limits_{j = 1, j \ne i}^{N} \pi^j(x, a^j)$.
Then, the (Markovian) transition probability from state $x \in \S$ to state $y \in \S$ when each agent $i$ plays according to its randomized strategy $\pi^i$ can be written as:
$$p(y | x, \pi) = \sum\limits_{a \in \A(x)}  p(y | x, a) \pi(x, a).$$
\item[\textbf{Reward:}]
$r^i(x,a)$ is the single-stage reward obtained by agent $i$ in state $x \in \S$, where $a \in \A(x)$ is the aggregate action taken.
\end{description}


\begin{definition}\label{def:vf}\textbf{(Value function)}
The value function is the expected return for any agent $i \in \{1, 2, \dots, N\}$ and is defined as
\begin{align}
\label{eq:games:sg-value}
v^i_\pi(s_0) = \E \left [\sum\limits_{t} \beta^t \sum_{a \in \A(x)} \left ( r^i(s_t, a) \pi(s_t,a) \right ) \right ].
\end{align}
\end{definition}
Given the above notion of the value function, the goal of each agent is to find a strategy that achieves a {\em Nash equilibrium}. The latter is defined as follows:
\begin{definition}\textbf{(Nash Equilibrium)}
\label{def:nash}
A stationary Markov strategy $\pi^* = \left < \pi^{1*}, \pi^{2*}, \dots, \pi^{N*} \right >$ is said to be Nash if
\[v^i_{\pi^*}(s) \ge v^i_{\left <\pi^i, \pi^{-i*} \right>}(s), \forall \pi^i, \forall i, \forall s \in \S.\] 
The corresponding equilibrium of the game is said to be Nash equilibrium.
\end{definition}
Since we consider a discounted stochastic game with a finite state space, we have the following well-known result that ensures the existence of stationary equilibrium:
  \begin{theorem}
  \label{thm:nash-exist}
    Any finite discounted stochastic game has an equilibrium in stationary strategies.
  \end{theorem}  
  We shall refer to such stationary randomized strategies as {\em Nash strategies}. The reader is referred to \cite{fink1964equilibrium},  \cite{takahashi1964equilibrium}, \cite{sobel1971noncooperative} for a proof of Theorem \ref{thm:nash-exist}. 

\section{A Generalized Optimization Problem}
\label{sect:problem}

\paragraph{Basic idea.}
Using dynamic programming the Nash equilibrium condition in Definition \ref{def:nash} can be re-written as: $\forall x \in \S, \forall i = 1, 2, \dots, N,$
\begin{equation}
\label{eq:opt:basic-dp}
v^i_{\pi^*}(x) = \max\limits_{\pi^i(x) \in \Delta(\A^i(x))} \left \{ E_{\pi^i(x)} Q^i_{\pi^{-i*}}(x, a^i) \right \},
\end{equation}
where 
\[Q^i_{\pi^{-i}}(x, a^i) = E_{\pi^{-i}(x)} \left [r^i (x, a) + \beta \sum \limits_{y \in U(x)} p(y|x, a) v^i(y) \right ],\]
represents the marginal value associated with picking action $a^i \in \A^i(x)$, in state $x \in \S$ for agent $i$, while other agents act according to $\pi^{-i}$. Also, $\Delta(\A^i(x))$ denotes the set of all possible probability distributions over $\A^i(x)$.

The basic idea behind the optimization problem that we formulate below is to model the objective such that the value function is correct w.r.t. the agents' strategies, while adding a constraint to ensure that a feasible solution to the problem corresponds to Nash equilibrium.
 
\paragraph{Objective.}
A possible optimization objective for agent $i$ would be
\[f^i(\v^i, \pi) = \sum\limits_{x \in \S} \left (v^i(x) - E_{\pi^i} Q^i_{\pi^{-i}}(x, a^i) \right ).\]
The objective above has to be minimized over all policies $\pi^i \in \Delta(\A^i(x))$. But $Q^i_{\pi^{-i}}(x, a^i)$, by definition, is dependent on strategies of all other agents. So, an isolated minimization of $f^i(\v^i, \pi^i)$ would not be meaningful and hence, we consider the aggregate objective $f(\v, \pi) = \sum\limits_{i = 1}^N f^i(\v^i, \pi)$. This objective
 is minimized over all policies $\pi^i \in \Delta(\A^i(x))$ of all agents. Thus, we have an optimization problem with objective as $f(\v, \pi)$ along with the natural constraints ensuring that the policy vectors $\pi^i(x)$ remain as probabilities over all feasible actions $\A^i(x)$ in all states  $x \in \S$ and for agents $i=1,\ldots,N$. 

\paragraph{Constraints.} 
Notice that an optimization problem with the objective discussed above has only a set of simple constraints ensuring that $\pi$ remains a valid strategy. However, this is not sufficient to accurately represent Nash equilibria of the underlying game. 
Here, we look at a possible set of additional constraints which might make the optimization problem more useful. Note that the term being maximized in equation \eqref{eq:opt:basic-dp}, i.e., $E_{\pi^i} Q^i_{\pi^{-i}}(x, a^i)$, represents a convex combination of the values of $Q^i_{\pi^{-i}}(x, a^i)$ over all possible actions $a^i \in \A^i(x)$ in a given state $x \in \S$ for a given agent $i$. Thus, it is implicitly implied that \[Q^i_{\pi^{-i}}(x, a^i) \le v^i_{\pi^*}(x), \forall a^i \in \A^i(x), x \in \S, i = 1, 2, \dots, N.\]
Formally, the optimization problem for any $N\ge2$ is given below:
\begin{equation}
\label{eqn:OP}
\left .
\begin{array}{l}
\min \limits_{\v,\pi} f(\v, \pi) = \sum \limits_{i = 1}^{N} \sum\limits_{x \in \S} \left (v^i(x) - E_{\pi^i} Q^i_{\pi^{-i}}(x, a^i) \right ) \text{s.t.} \\[2ex]
\subequationitem\label{subeq:opt:basic-idea:first-cut-opt:pi-ge-0} \pi^i(x, a^i) \ge 0, \forall a^i \in \A^i(x), x \in \S, i = 1, 2, \dots, N,\\
\subequationitem\label{subeq:opt:basic-idea:first-cut-opt:pi-sum-1} \sum \limits_{i = 1}^N \pi^i(x, a^i) = 1, \forall x \in \S, i = 1, 2, \dots, N,\\
\subequationitem\label{subeq:opt:basic-idea:first-cut-opt:pi-ne-2} Q^i_{\pi^{-i}}(x, a^i) \le v^i(x), \forall a^i \in \A^i(x), x \in \S, i = 1, 2, \dots, N.
\end{array} \right \}
\end{equation}
In the above, \ref{subeq:opt:basic-idea:first-cut-opt:pi-ge-0}--\ref{subeq:opt:basic-idea:first-cut-opt:pi-sum-1} ensure that $\pi$ is a valid policy, while \ref{subeq:opt:basic-idea:first-cut-opt:pi-ne-2} is necessary for any valid policy to be a NE of the underlying game.

\begin{theorem}
\label{theorem:f-equal-zero-nash}
A feasible point $(v^*, \pi^*)$ of the optimization problem (\ref{eqn:OP}) provides a Nash equilibrium in stationary strategies to the corresponding general-sum discounted stochastic game if and only if $f(v^*, \pi^*) = 0$.
\end{theorem}
\begin{proof}
See \citep[Theorem 3.8.2]{filar-vrieze} for a proof in the case of $N = 2$. The proof works in a similar manner for general $N$.
\end{proof}

A few remarks about the difficulties involved in solving \eqref{eqn:OP} are in order.
\begin{remark}\textbf{\textit{(Non-linearity)}}
For the case of $N=2$, the objective $f(\v,\pi)$ in \eqref{eqn:OP} is of order $3$, while the toughest constraint \ref{subeq:opt:basic-idea:first-cut-opt:pi-ne-2} is quadratic. This is apparent from the fact that the second term inside the summation in $f(\v,\pi)$ has the following variables multiplied: $\pi_1$ in the first expectation, $\pi_2$ inside the expectation in the definition of the $Q$-function and $v$ inside the second term of the expectation in $Q$-function. Along similar lines, the constraint \ref{subeq:opt:basic-idea:first-cut-opt:pi-ne-2} can be inferred to be quadratic. Thus, we have an optimization problem with a third-order objective and quadratic constraints, even for the case of $N=2$ and the constituent functions (both objective and constraints) can be easily seen to be non-linear. For a general $N>2$, the problem \eqref{eqn:OP} gets more complicated, as more policy variables $\pi_1,\ldots, \pi_N$ are thrown in. 
\end{remark}

\begin{remark}\textbf{\textit{(Beyond local minima)}}
 \cite{filar-vrieze} have formulated a non-linear optimization problem for the case of two-player {\em zero-sum} stochastic games. An associated result (Theorem 3.9.4, page 141, in \citep{filar-vrieze}) states that every local minimum of that optimization problem is also a global minimum. However, this result does not hold for a \textit{general-sum} game even in the two-player (and also $N\ge 3$) setting and hence, the requirement is for a global optimization scheme that solves \eqref{eqn:OP}. 
\end{remark}

\begin{remark}\textbf{\textit{(No steepest descent)}}
From the previous remark, it is apparent that a simple {\em steepest descent} scheme is not enough to solve \eqref{eqn:OP} even for the two-player setting. This is because there can be local minima of \eqref{eqn:OP} that do not correspond to the global minimum and steepest descent schemes guarantee convergence to local minima only. 
Note that steepest descent schemes were sufficient to solve for Nash equilibrium strategies in two-player zero-sum stochastic games, while this is not the case with general-sum $N$-player games, with $N\ge 2$. Sections 3 and 4 of \citep{hlpthesis} contain a detailed discussion on inadequacies of steepest descent for general-sum games.
\end{remark}

\begin{remark}\textbf{\textit{(No Newton method)}}
A natural question that arises is can one employ a Newton method in order to solve \eqref{eqn:OP} and the answer is in the negative. Observe that the Hessian of the objective function $f(\v,\pi)$ in \eqref{eqn:OP} has its diagonal elements to be zero and this makes it \textit{indefinite}. This make Newton methods infeasible as they require invertibility of the Hessian to work.
\end{remark}
 
We overcome the above difficulties by deriving a descent direction (that is not necessarily steepest) in order to solve \eqref{eqn:OP}. Before we present the descent direction, we break-down \eqref{eqn:OP} into simpler sub-problems. Subsequently we descibe Stochastic Game - Sub-Problem (SG-SP) conditions that are both necessary and sufficient for the problem \eqref{eqn:OP}.

\subsection*{Sub-problems for each state and agent}
We form sub-problems from the main optimization problem (\ref{eqn:OP}) along the lines of \cite{kktsp}, for each state $x \in \S$ and each agent $i \in \{1, 2, \dots, N\}$. The sub-problems are formed with the objective of ensuring that there is no Bellman error (see $g^i_{x, z}(\theta)$ below). 
For any $ x \in \S$, $z = 1, 2, \dots, |\mathcal{A}^i(x)|$ and $i \in \{ 1, 2, \dots, N\}$, let 
$\theta :=  \left <\v^i, \pi^{-i}(x)\right>$ denote the value-policy tuple and let
\begin{align}
g^i_{x, z}(\theta)  :=  Q^i_{\pi^{-i}}(x, a^i_z) - v^i(x) \label{eq:g}
\end{align}
denote the Bellman error. Further, let $p_z  :=  \pi^i(x, a^i_z)$ and $p  = \left <p_z: z = 1, 2, \dots, |\mathcal{A}^i(x)|\right>$. 
Then, the sub-problems are formulated as follows: 
\begin{align}
\label{eqn:OP-term-part}
&\min \limits_{\theta, p} h_x(\theta, p) = \sum \limits_{z = 1}^{|\mathcal{A}^i(x)|} p_z\left [ - g^i_{x, z}(\theta) \right ] \vspace{1ex}\\
&\text{s.t. }  
g^i_{x, z}(\theta) \leq 0, -p_z \leq 0, \text{ for } z = 1, 2, \dots, |\mathcal{A}^i(x)|,\nonumber\\
&\text{ and } \sum\limits_z p_z=1.\nonumber
\end{align}

\section{Stochastic Game - Sub-Problem (SG-SP) Conditions}
\label{sec:sgsp}
In this section, we derive a set of necessary and sufficient conditions for solutions of \eqref{eqn:OP} and establish their equivalence with Nash strategies. 
\begin{definition}[SG-SP Point]
\label{definition:sg-sp}
A point $(\v^*, \pi^*)$ of the optimization problem (\ref{eqn:OP}) is said to be an SG-SP point if it is a feasible point of (\ref{eqn:OP}) and for every sub-problem, i.e., for all $ x \in \S$ and $i \in \{1, 2, \dots, N\}$,
\begin{equation}
\label{eqn:sg-sp}
p_z^* g^i_{x, z}(\theta^*) = 0, \qquad \forall z = 1, 2, \dots, |\mathcal{A}^i(x)|.
\end{equation}
\end{definition}
The above conditions, which define a point to be an SG-SP point, are called SG-SP conditions.

\subsection{Equivalence of SG-SP with Nash strategies}
The connection between SG-SP points and Nash equilibria  can be seen intuitively as follows:\\ 
\begin{inparaenum}[\bfseries(i)]
\item The objective function $f(v^*, \pi^*)$ in (\ref{eqn:OP}) can be expressed as a summation of terms of the form $p_z^* [ - g^i_{x, z}(\theta^*)]$ over $z = 1, 2, \dots, |\mathcal{A}^i(x)|$ and over all sub-problems. Condition (\ref{eqn:sg-sp}) suggests that each of these terms is zero which implies $f(v^*, \pi^*) = 0$. \\
\item The objective of the sub-problem is to ensure that there is no Bellman error, which in turn implies that the value estimates $v^*$ are correct with respect to the policy $\pi^*$ of all agents. \\
\end{inparaenum}
Combining (i) and (ii) with Theorem 3.8.2 of \citep{filar-vrieze}, we have the following result: 
\begin{theorem}[Nash $\Leftrightarrow$ SG-SP]
\label{theorem:nash-sg-sp}
A strategy $\pi^*$ is Nash if and only if $(v^*, \pi^*)$ for the corresponding optimization problem (\ref{eqn:OP}) is an SG-SP point.
\end{theorem}
The proof of the above theorem follows from a combination of Lemmas \ref{lemma:sg-sp-implies-nash} and \ref{lemma:nash-implies-sg-sp}, presented below.
\begin{lemma}[SG-SP $\Rightarrow$ Nash]
\label{lemma:sg-sp-implies-nash}
An SG-SP point $(\v^*, \pi^*)$ gives Nash strategy-tuple for the underlying stochastic game.
\end{lemma}
\begin{proof}
The objective function value $f(v^*, \pi^*)$ of the optimization problem (\ref{eqn:OP}) can be expressed as a summation of terms of the form $p_z^* [ - g^i_{x, z}(\theta^*)]$ over $z = 1, 2, \dots, m$ and over all sub-problems. Condition (\ref{eqn:sg-sp}) suggests that each of these terms is zero which implies $f(v^*, \pi^*) = 0$. From \citet[Theorem 3.8.2, page 132]{filar-vrieze}, since $(v^*, \pi^*)$ is a feasible point of (\ref{eqn:OP}) and $f(v^*, \pi^*) = 0$, $(v^*, \pi^*)$ corresponds to Nash equilibrium of the underlying stochastic game.
\end{proof}
\begin{lemma}[Nash $\Rightarrow$ SG-SP]
\label{lemma:nash-implies-sg-sp}
A strategy $\pi^*$ is Nash if $(v^*, \pi^*)$ for the corresponding optimization problem (\ref{eqn:OP})  is an SG-SP point.
\end{lemma}
\begin{proof}
From \citet[Theorem 3.8.2, page 132]{filar-vrieze}, if a strategy $\pi^*$ is Nash, then a feasible point $(v^*, \pi^*)$ exists for the corresponding optimization problem (\ref{eqn:OP}), where $f(v^*, \pi^*) = 0$. From the constraints of (\ref{eqn:OP}), it is clear that for a feasible point, $p_z^* [ - g^i_{x, z}(\theta^*)] \ge 0$, for $z = 1, 2, \dots, m$, for every sub-problem. Since the sum of all these terms, i.e., $f(v^*, \pi^*)$, is zero, each of these terms is zero, i.e., $(v^*, \pi^*)$ satisfies (\ref{eqn:sg-sp}). Thus, $(v^*, \pi^*)$ is an SG-SP point.
\end{proof}

\subsection{Kinship to Karush-Kuhn-Tucker - Sub-Problem (KKT-SP) conditions}
\label{sec:appendix-kktsp}
  \cite{kktsp} consider a similar optimization problem as \eqref{eqn:OP} for the case of two agents, i.e., $N=2$ and derive a set of verifiable necessary and sufficient conditions that they call KKT-SP conditions. In the following, we first extend the KKT-SP conditions to $N$-player stochastic games, for any $N\ge 2$ and later establish the equivalence between KKT-SP conditions with the SG-SP conditions formulated above.

The Lagrangian corresponding to (\ref{eqn:OP-term-part}) can be written as
\begin{align}
k(\theta,p,\lambda,\delta, s, t) = & h_x(\theta,p) + \sum_{z = 1}^{|\mathcal{A}^i(x)|} \lambda_z \left ( g^i_{x, z}(\theta) + s_z^2 \right )  + \sum_{z = 1}^{|\mathcal{A}^i(x)|} \delta_z \left ( - p_z + t_z^2 \right ),\label{eqn:lagrangian}
\end{align}
where $\lambda_z$ and $\delta_z$ are the Lagrange multipliers and $s_z$ and $t_z, z= 1, 2, \dots, |\mathcal{A}^i(x)|$ are the slack variables, corresponding to the first and second constraints of the sub-problem \eqref{eqn:OP-term-part}, respectively. 

Using the Lagrangian \eqref{eqn:lagrangian}, the associated KKT conditions for the sub-problem (\ref{eqn:OP-term-part}) corresponding to a state $x \in S$ and agent $i \in \{1,\ldots,N\}$ at a point $\left <\theta^*, p^*\right >$ are the following:
\begin{equation}
\label{eqn:kkt-sp}
\left .
\begin{array}{l}
\subequationitem\label{subeq:kkt1} \nabla_\theta h_x(\theta^*, p^*) + \sum \limits_{z = 1}^{m} \lambda_{z}\nabla_\theta g^i_{x, z}(\theta^*) = 0,\vspace{0.5em}\\
\subequationitem\label{subeq:kkt2} \dfrac{\partial h_x(\theta^*, p^*)}{\partial p_z} - \delta_z + \delta_m = 0, \qquad z = 1, 2, \dots, m,\vspace{0.5em}\\
\subequationitem\label{subeq:kkt3} \delta_z p_z^* = 0, \qquad z = 1, 2, \dots, m,\vspace{0.5em}\\
\subequationitem\label{subeq:kkt4} \lambda_{z} g^i_{x, z}(\theta^*) = 0, \qquad z = 1, 2, \dots, m,\vspace{0.5em}\\
\subequationitem\label{subeq:kkt5} \lambda_{z} \geq 0, \qquad z = 1, 2, \dots, m,\vspace{0.5em}\\
\subequationitem\label{subeq:kkt6} \delta_z \geq 0, \qquad z = 1, 2, \dots, m.
\end{array}
\right \}
\end{equation}
KKT-SP conditions are shown to be necessary and sufficient for $(\v^*, \pi^*)$ to represent a Nash equilibrium point of the underlying stochastic game and $\pi^*$ to be a Nash strategy-tuple in the case of $N = 2$ (see Theorem 3.8 in \citep{kktsp}). However, this requires an additional assumption that for each sub-problem, $\left \{\nabla_\theta g^i_{x, z}(\theta^*): z = 1, 2, \dots, m\right \}$ is a set of linearly independent vectors. On the other hand, the SG-SP conditions (see Definition \ref{definition:sg-sp}) that we formulate do not impose any additional linear independence requirement, in order to ensure that the solution points of the sub-problems correspond to Nash equilibria. 

The following lemma establishes the kinship between SG-SP and KKT-SP conditions. 
\begin{lemma}[KKT-SP $\Rightarrow$ SG-SP]
A KKT-SP point is also an SG-SP point.
\end{lemma}
\begin{proof}
A KKT-SP point $(\v^*, \pi^*)$ is a feasible point of (\ref{eqn:OP}). For every sub-problem, substitute and eliminate $\lambda_z^* = p_z^*$ and $\delta_z^* = -g^i_{x, z}(\theta^*)$, $z = 1, 2, \dots, m$. Then
\begin{enumerate}
\item Conditions (\ref{subeq:kkt1}) and (\ref{subeq:kkt2}) are satisfied;
\item Conditions (\ref{subeq:kkt3}) and (\ref{subeq:kkt4}) reduce to (\ref{eqn:sg-sp}); and
\item Conditions (\ref{subeq:kkt5}) and (\ref{subeq:kkt6}) are satisfied as the point $(\v^*, \pi^*)$ is assumed to be feasible.
\end{enumerate}
\end{proof}

From the above, it is evident that the simpler and more general (for any $N$) SG-SP conditions can be used for Nash equilibria  as compared to KKT-SP conditions because:\\
\begin{inparaenum}[\bfseries(i)]
 \item every KKT-SP point is also an SG-SP point and \\
 \item SG-SP conditions do not impose any additional linear independence requirement in order to be Nash points.
\end{inparaenum}

\section{OFF-SGSP: Offline, Model-Based}
\label{sect:off-sgsp}
\paragraph{Basic idea.}
As outlined in the introduction, OFF-SGSP is an actor-critic algorithm that operates using two timescale recursions as follows \\
\begin{description}
\item[Critic recursion:] This estimates the value function $\v$ using value iteration, along the faster timescale; and\\
\item[Actor recursion:] This operates along the slower timescale and updates the policy in the descent direction so as to ensure convergence to an SG-SP point.\\ 
\end{description}
As mentioned before, OFF-SGSP is a model-based algorithm and the transition dynamics and reward structure of the underlying game are used for both steps above. 

\paragraph{Update rule.}
Using two timescale recursions, OFF-SGSP updates the value-policy tuple $(v,\pi)$ as follows: For all $x \in \S$ and  $a^i \in \A^i(x)$,
\begin{align}
\label{eqn:off-sgsp-pi}
\textbf{Actor: }&\pi^i_{n + 1}(x, a^i) = \Gamma \left(\pi^i_{n}(x, a^i) - b(n) \sqrt{\pi^i_{n}(x, a^i)} \left|g^i_{x, a^i}(\v^i_n, \pi^{-i}_n)\right| \bsgn\left(\dfrac{\partial f(\v_n, \pi_n)}{\partial \pi^i}\right)\right),\\
\label{eqn:off-sgsp-v}
\textbf{Critic: }&v^i_{n + 1}(x) \!=\! v^i_n(x) + c(n) \sum\limits_{a^i \in \A^i(x)} \pi^i_{n}(x, a^i) g^i_{x, a^i}(\v^i_n, \pi^{-i}_n),
\end{align}
where 
$g^i_{x, a^i}(\v^i, \pi^{-i}):=Q^i_{\pi^{-i}}(x, a^i) -v^i(x)$ denotes the Bellman error,
$f(\v, \pi)$ is the objective function in \eqref{eqn:OP} and $\Gamma$ is a projection operator that ensures that the updates to $\pi$ stay within the simplex $\mathcal{D} = \{ (d_2,\ldots,d_{|\A^i(x)|}) \mid d_i \ge 0, \forall i=1,\ldots,|\A^i(x)|, \sum\limits_{j=2}^{|\A^i(x)|} d_j \le 1\}$. 
Further, using $\Gamma$, one ensures that $d_1 = 1- \sum_{j=2}^{|\A^i(x)|}, d_j \in [0,1]$.
Here $\bsgn(\cdot)$ is a continuous version of the sign function and projects any $x$, outside of a very small interval around $0$, to $\pm 1$ according to the sign of $x$ (see Remark \ref{remark:sign} for a precise definition). Continuity of $\bsgn(\cdot)$ is a technical requirement that helps in providing strong convergence guarantees\footnote{
Using the normal $\sgn()$ function is problematic for an ODE approach based analysis, as $\sgn()$ is discontinuous. In other words, with $\sgn()$ the underlying system for the policy $\pi_n^i$ will be a stochastic recursive inclusion and providing meaningful guarantees for such inclusions would require more assumptions. In comparison, the results are stronger for the ODE approach.  This is the motivation behind employing $\bsgn()$, which is a continuous extension of $\sgn()$. The function $\bsgn(x)$ projects any $x$ outside of a small interval around $0$ (say $[-\nu,\nu]$ for some $\nu>0$ small) to either $+1$ or $-1$ as $\sgn()$ would do and within the interval $[-\nu,\nu]$, one may choose $\bsgn(x)=x$ or any other continuous function with compatible end-point values. One could choose $\nu$ arbitrarily close to $0$, making $\bsgn$ practically very close to $sgn$.}.

The following assumption on the step-sizes ensures that the $\pi$-recursion \eqref{eqn:off-sgsp-pi} proceeds on a slower timescale in comparison to the $v$-recursion \eqref{eqn:off-sgsp-v}:
\begin{assumption}
\label{assumption:step-sizes}
The step-sizes $\{b(n)\}, \{c(n)\}$ satisfy
\begin{align*}
&\sum\limits_{n = 1}^\infty b(n) \!=\!\sum\limits_{n = 1}^\infty c(n) \!=\! \infty,
\sum\limits_{n = 1}^\infty \left ( b^2(n) + c^2(n) \right ) \!<\! \infty, \dfrac{b(n)}{c(n)} \rightarrow 0.
\end{align*}
\end{assumption}

\paragraph{Justification for descent direction.}
The following proposition proves that the decrement for the policy in \eqref{eqn:off-sgsp-pi} is a valid descent direction for the objective function $f(\cdot,\cdot)$ in \eqref{eqn:OP}.
\begin{proposition}
\label{proposition:descent-direction}
For each $i = 1, 2, \dots, N, x \in \S, a^i \in \A^i(x)$, we have that \\$-\sqrt{\pi^i(x, a^i)} \left|g^i_{x, a^i}(\v^i, \pi^{-i})\right|\bsgn\left(\dfrac{\partial f(\v, \pi)}{\partial \pi^i}\right)$ is a non-ascent, and in particular a descent direction \\if $\sqrt{\pi^i(x, a^i)} \left|g^i_{x, a^i}(\v^i, \pi^{-i})\right| \ne 0$, in the objective $f(\v,\pi)$ of (\ref{eqn:OP}). 
\end{proposition}
\begin{proof}
The objective $f(\v,\pi)$ can be rewritten as 
\begin{align*}
 f(\v, \pi) = \sum\limits_{i = 1}^N \sum\limits_{x \in \S} \sum\limits_{a^i \in \A^i(x)} \left \{ \pi^i(x, a^i) \left [- g^i_{x, a^i}(\v^i, \pi^{-i}) \right ] \right \}.
\end{align*}
Assume $f(\v, \pi) > 0$, otherwise the solution to \eqref{eqn:OP} is already achieved.
For an $a^i \in \A^i(x)$ for some $x \in \S$ and $i \in \{1, 2, \dots, N\}$, let
{\small
\[\hat \pi^i(x, a^i) \!=\! \pi^i(x, a^i) - \delta \sqrt{\pi^i(x, a^i)} \left|g^i_{x, a^i}(\v^i, \pi^{-i})\right|\bsgn\left(\dfrac{\partial f(\v, \pi)}{\partial \pi^i}\right),\]} for a small $\delta > 0$.  Let $\hat \pi$ be the same as $\pi$ except that action $a^i$ is picked as defined above. Then by a Taylor series expansion of $f(\v, \hat \pi)$ till the first order term, we obtain 
\begin{align*}
f(\v, \hat \pi) = f(v, \pi) + \delta \left[-\sqrt{\pi^i(x, a^i)} \left|g^i_{x, a^i}(\v^i, \pi^{-i})\right|\right] \bsgn\left(\dfrac{\partial f(\v, \pi)}{\partial \pi^i}\right)\dfrac{\partial f(\v, \pi)}{\partial \pi^i(x, a^i)} + o(\delta).
\end{align*}
The rest of the proof amounts to showing that the second term in the expansion above is $\le 0$.
This can be inferred as follows:
\begin{align*}
&-\sqrt{\pi^i(x, a^i)} \left|g^i_{x, a^i}(\v^i, \pi^{-i})\right| \bsgn\left(\dfrac{\partial f(\v, \pi)}{\partial \pi^i}\right)\dfrac{\partial f(\v, \pi)}{\partial \pi^i} \\
=& -\sqrt{\pi^i(x, a^i)} \left|g^i_{x, a^i}(\v^i, \pi^{-i})\right|\left| \dfrac{\partial f(\v, \pi)}{\partial \pi^i}\right| \le 0,
\end{align*}
and is in particular $<0$ if $\sqrt{\pi^i(x, a^i)} \left|g^i_{x, a^i}(\v^i, \pi^{-i})\right|\ne0$. 

Thus, for $a^i \in \A^i(x)$, $x \in \S$ and $i \in \{1, 2, \dots, N\}$ where $\pi^i(x, a^i) > 0$ and $g^i_{x, a^i}(\v^i, \pi^{-i}) \ne 0$, $f(\v, \hat \pi) < f(\v, \pi)$ for small enough $\delta$.
The claim follows.
\end{proof}

\paragraph{Main result.}
Let $R^i_\pi = \left < r^i(x, \pi), x \in \S \right >$ be a column vector of rewards to agent $i$ and $P_\pi = [ p(y|x, \pi), x \in \S, y \in \S ]$ be the transition probability matrix, both for a given $\pi$. Then, the value function for a given policy $\pi$ is defined as
\begin{equation}
\label{eqn:v_pi-main}
\v^i_\pi = \left [ I - \beta P_\pi \right ]^{-1} R^i_\pi, i = 1, 2, \dots, N.
\end{equation}
The above will be used to characterize the limit of the critic-recursion \eqref{eqn:off-sgsp-v}. 
Before presenting the main result, we specify the ODE that underlies the actor-recursion \eqref{eqn:off-sgsp-pi}:
For all $ a^i \in \A^i(x), x \in \S, i = 1, 2, \dots, N,$
\begin{align}
\label{eqn:pi-ode-main}
\dfrac{d \pi^i(x, a^i)}{dt} = \bar\Gamma \left ( -\sqrt{\pi^i(x, a^i)} \left|g^i_{x, a^i}(\v_\pi^i, \pi^{-i})\right| \bsgn\left(\dfrac{\partial f(\v_\pi, \pi)}{\partial \pi^i}\right)\right ), 
\end{align}
where $\bar\Gamma$ is a projection operator that restricts the evolution of the above ODE to the simplex $\mathcal{D}$ (see Section \ref{sect:convergence} for a precise definition).

The main result regarding the convergence of OFF-SGSP is as follows:
\begin{theorem}
\label{thm:off-sgsp-main}
Let $G$ denote the set of all feasible points of the optimization problem \eqref{eqn:OP} and $K$ the set of limit points of the ODE \eqref{eqn:pi-ode-main}. Further, 
let $K_1 = K \cap G$ and $\K^* = \{ (\v^i_{\pi^*}, \pi^*) \mid \pi^* \in K_1\}$. 
Then, for any agent $i=1,\ldots,N$, the sequence of iterates $(v^i_{n}, \pi^i_{n}), n\ge 0$ satisfy 
$$ (v^i_{n}, \pi^i_{n})  \rightarrow \K^* \text{ a.s.}$$
\end{theorem}
\begin{proof}
See Section \ref{sect:offsgsp-proof}.
\end{proof}
From the above theorem, we can infer the following:
\begin{enumerate}[\bfseries (i)]
\item The set of infeasible limit points of the ODE \eqref{eqn:pi-ode-main}, i.e., $K_2 = K \setminus K_1$, are asymptotically unstable (see Lemma \ref{lemma:pi-convergence} in Section \ref{sect:convergence} for a formal proof); and
\item OFF-SGSP converges to the set $\K^*$, which is the set of all asymptotically stable limit points of the system of ODEs (\ref{eqn:pi-ode}). Further, $\K^*$ corresponds to SG-SP (and hence Nash) points and hence, OFF-SGSP is shown to converge almost surely to a NE of the underlying discounted stochastic game.
\end{enumerate}
%
\section{ON-SGSP: Online and Model-Free}
\label{sect:on-sgsp}

\begin{figure}
\centering
\tikzstyle{block} = [draw, rectangle,  line width=0.2mm,join=round,minimum height=1.4cm, minimum width=4cm]
\tikzstyle{smallblock} = [draw, rectangle, minimum height=2em, minimum width=2em]
\begin{tikzpicture}
\node [block, rectangle, fill=green!30] (env) {{\bf\large Environment}};
\node [smallblock, rectangle, below=50pt of env, xshift=-10pt, fill=red!30] (2) {$\bm{2}$};
\node [smallblock, rectangle, left=80pt of 2, fill=red!30] (1) {$\bm{1}$};
\node [right=30pt of 2] (dots) {$\bm{\ldots}$};
\node [smallblock, rectangle, right=40pt of dots, fill=red!30] (N) {$\bm{N}$};
\node [below=40pt of 2] (don-sgsp) {{\bf\large ON-SGSP}};


\draw[>=latex',->,darkgreen] (env.180) to[out=-180,in=90] node [left] {$\bm{r^1}$, $\bm{y}$} (1);
\draw[>=latex',->,red] (1) -- node[left] {$\bm{a^1}$} (env.190);

\draw[>=latex',->,darkgreen] (env.south) to node [left] {$\bm{r^2}$, $\bm{y}$} (2);
\draw[>=latex',->,red] (2) to [out=45,in=-45] node[right] {$\bm{a^2}$} (env);

\draw[>=latex',->,darkgreen] (env) to[out=0,in=90] node [right] {$\bm{r^N}$, $\bm{y}$} (N);
\draw[>=latex',->,red] (N) -- node[left, near start] {$\bm{a^N}$} (env.350);

\draw[>=latex',->,bleu2] (don-sgsp) -- (1.south);
\draw[>=latex',->,bleu2] (don-sgsp) -- (2.south);
\draw[>=latex',->,bleu2] (don-sgsp) -- (N.south);
\end{tikzpicture}
\caption{ON-SGSP's decentralized on-line learning model with $N$ agents}
\label{fig:algos:don-sgsp}
\end{figure}
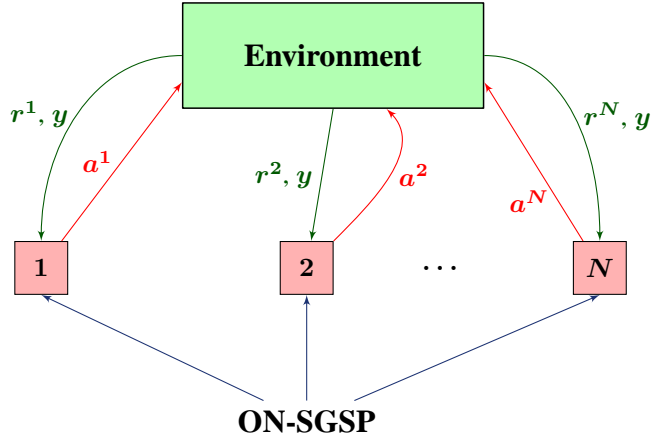
Though OFF-SGSP is suitable for only off-line learning of Nash strategies, it is amenable for extension to the general (on-line) multi-agent RL setting where neither the transition probability $p$ nor the reward function $r$ are explicitly known. ON-SGSP operates in the latter model-free setting and uses the stochastic game as a generative model.  

As illustrated in Fig. \ref{fig:algos:don-sgsp}, every iteration in ON-SGSP represents a discrete-time interaction with the environment, where each agent presents its action to the environment and observes the next state and the reward vector of all agents. The learning is localized to each agent $i~\in~\{ 1, 2, \dots, N \}$, making the setting decentralized. This is in the spirit of earlier multi-agent RL approaches (cf. \cite{huwellman}, \cite{huwellman2003} and \cite{littman}).

Algorithm~\ref{algo:drl-sqrt-pi-v} presents the complete structure of ON-SGSP along with update rules for the value and policy parameters. 
The algorithm operates along two timescales as follows:
\begin{description}
 \item[Critic (faster timescale):] Each agent estimates its own value function as well as that of other agents, using a temporal-difference (TD) \citep{sutton1988learning} type update in \eqref{eqn:don-sgsp-v}. Moreover, the gradient $\dfrac{\partial f(\v_n, \pi_n)}{\partial \pi^i(x, a^i)}$ is also estimated in an online manner via the $\xi$-recursion in \eqref{eqn:don-sgsp-q}. Note that the $\xi$-recursion is made necessary due to the fact that ON-SGSP operates in a model-free setting.
 \item[Actor (slower timescale):] The policy update is similar to OFF-SGSP, except that the estimates of value $v$ and gradient $\xi$ are used to derive the decrement in \eqref{eqn:don-sgsp-pi}.
\end{description}
Note that, since ON-SGSP operates in a model-free setting, both the value and policy updates are different in comparison to OFF-SGSP. 
The value $\v$ update \eqref{eqn:don-sgsp-v} on the faster timescale can be seen to be the stochastic approximation variant of value iteration and it converges to the same limit as in OFF-SGSP, without knowing the model. 
On the other hand, the policy update \eqref{eqn:don-sgsp-pi} on the slower timescale involves a decrement that is motivated by the descent direction suggested by Proposition \ref{proposition:descent-direction}.  

\begin{algorithm}[h]
\caption{ON-SGSP}
\label{algo:drl-sqrt-pi-v}
\begin{algorithmic}
\State {\bfseries Input:} Starting state $x_0$, initial point $\theta^i_0 = (\v^i_0, \pi^i_0)$,  step-sizes $\{b(n), c(n)\}_{n \ge 1}$, number of iterations to run $M >> 0$. 
\State {\bfseries Initialization:} $n \leftarrow 1$, $\theta^i \leftarrow \theta^i_0$, $x \leftarrow x_0$
\For {$n=1,\dots,M$}
\State Play action $a^i_n:=\pi^i_n(x_n)$ along with other agents in current state $x_n \in \S$
\State Obtain next state $y_n \in \S$
\State Observe reward vector $r_n=<r^1_n,\ldots,r^N_n>$
\begin{align}
&\textrm{{\bf Value Update:  } For } j=1,\dots,N \phantom{.}\nonumber\\[-1ex] 
&v^j_{n+1}(x_n) \!=\! v^j_n(x_n) \!+\! c(n) \left( r^j_n \!+\! \beta v^j_n(y_n) \!- \!v^j_n(x_n)\right) \label{eqn:don-sgsp-v}\\[0.5ex]
&\textrm{\bf Gradient Estimation:} \phantom{.}\nonumber\\[-1.5ex]  \label{eqn:don-sgsp-q}
&\xi^i_{n+1}(x_n, a^i_n) = \xi^i_{n}(x_n, a^i_n) + c(n) \bigg( \sum\limits_{j=1}^N \big( r^j_n + \beta v^j_n(y_n) - v^j_n(x_n) \big) -  \xi^i_n(x_n, a^i_n) \bigg)
 \\[-1ex]
&\textrm{{\bf Policy Update:}}\phantom{.}\nonumber\\[-1ex] 
&\pi^i_{n+1}(x_n, a^i_n) = \Gamma (\pi_n^i(x_n, a^i_n) - b(n)\sqrt{\pi_n^i(x_n, a_n^i)}\nonumber\\
&\qquad\qquad\qquad\times  \left|r_n^i + \beta v_n^i(y_n) - v_n^i(x_n)\right| \bsgn(-\xi_{n+1}^i(x_n, a^i_n)))\label{eqn:don-sgsp-pi}
\end{align}
\EndFor
\end{algorithmic}
\end{algorithm}

 A few remarks about ON-SGSP are in order.
 \begin{remark}
 \textbf{(Coupled dynamics)}
In the ON-SGSP algorithm, an agent $i$ observes the rewards of other agents and uses this information to compute the respective value estimates. These quantities are then used to derive the decrement in the policy update \eqref{eqn:don-sgsp-pi}. This is meaningful in the light of the impossibility result of \cite{hart2003uncoupled}, where the authors show that in order to converge to a Nash equilibrium each agent's strategy needs to factor in the rewards of the other agents. 
 \end{remark}
\begin{remark}
\label{remark:complexity}
\textbf{(Per-iteration complexity)}
\begin{description}
 \item[OFF-SGSP:] Let $A$ be the typical number of actions available to any agent in any given state and let $U$ be the typical number of next states for any state $x \in \S$. Then, the typical number of multiplications in OFF-SGSP per iteration is $N \times \left ( ( U + 1 ) \times A^N + 4 A \right ) \times |\S|$. Thus, the computational complexity grows exponentially in terms of the number of agents while being linear in the size of the state space. Note that the exponential behaviour in $N$ appears because of the computation of expectation over possible next states and strategies of agents. This computation is avoided in ON-SGSP.
 \item[ON-SGSP:] For each agent, a typical iteration would take just $(2 A + 1)$ number of multiplications.
\end{description}
Thus, per-iteration complexity of OFF-SGSP is $\Theta(2^N)$ while that of ON-SGSP is $\Theta(1)$ (from the point of view of each agent). Thus, ON-SGSP is computationally efficient and this is also confirmed by simulation results, which establish that the total run-time till convergence of ON-SGSP is indeed very small when compared to that of the off-line algorithm OFF-SGSP. In comparison, the stochastic tracing procedure of \citet{herings2004stationary} has a complexity of $O(|\S| \times A^N)$ per iteration which is similar to that of OFF-SGSP.

However, the per-iteration complexity alone is not sufficient and an analysis of the number of iterations required is necessary to complete the picture\footnote{A well-known complexity result \citep{papadimitriou1994complexity} establishes that finding the Nash equilibrium of a two-player game is PPAD-complete.}. On the other hand, convergence rate results for general multi-timescale stochastic approximation schemes are not available, see however, \cite{konda2004convergence} for rate results of two timescale schemes with linear recursions. 
\end{remark}

\begin{remark}
\textbf{(Descent directions)}
It is shown in Proposition \ref{proposition:descent-direction} that \\$-\sqrt{\pi^i(x, a^i)}\left|g^i_{x, a^i}(\v^i, \pi^{-i})\right|\sgn\left(\dfrac{\partial f(\v, \pi)}{\partial \pi^i}\right)$ is a descent direction in $\pi^i(x, a^i)$ for every $i = 1, 2, \dots, N, ~x \in \S, a^i \in \A^i(x)$. 
Since $\left \{\pi^i(x, a^i) \right \}^\alpha \ge 0$ for any $\alpha \ge 0$, 
\[-\left \{\pi^i(x, a^i) \right \}^\alpha \sqrt{\pi^i(x, a^i)}\left|g^i_{x, a^i}(\v^i, \pi^{-i})\right|\sgn\left(\dfrac{\partial f(\v, \pi)}{\partial \pi^i}\right)\] can also be seen to be a descent direction in $\pi^i(x, a^i)$ for every $i = 1, 2, \dots, N, ~ x \in \S, a^i \in \A^i(x)$. In other words,  the following is a descent direction in $\pi^i(x, a^i)$ for every $i = 1, 2, \dots, N, ~ x \in \S, a^i \in \A^i(x)$:
\[-\left \{\pi^i(x, a^i) \right \}^{\alpha'} \left|g^i_{x, a^i}(\v^i, \pi^{-i})\right|\sgn\left(\dfrac{\partial f(\v, \pi)}{\partial \pi^i}\right), \text{ for any } \alpha' \ge \frac{1}{2}\]  So, the policy updates in OFF-SGSP/ON-SGSP can be generalized as follows: With $\alpha \ge \frac{1}{2}$,
\begin{align*}
\text{\bf OFF-SGSP:} \quad & \pi^i(x, a^i) \leftarrow \Gamma \left(\pi^i(x, a^i) - \gamma \left \{\pi^i(x, a^i) \right \}^{\alpha'} \left|g^i_{x, a^i}(\v^i, \pi^{-i})\right|\bsgn\left(\dfrac{\partial f(\v, \pi)}{\partial \pi^i}\right) \right),\\
\text{\bf ON-SGSP:} \quad& \pi^i(x, a^i) \leftarrow \Gamma \left(\pi^i(x, a^i) - \gamma \left \{\pi^i(x, a^i) \right \}^{\alpha'} \left|r^i + \beta v^i(y) - v^i(x)\right| \bsgn(q^i(x, a^i) \right),
\end{align*} 
where $\gamma$ is a step-size parameter.
\end{remark}

\begin{remark}\textbf{(Convergence result)}
Theorem \ref{thm:off-sgsp-main} holds for ON-SGSP as well, while the proof deviates significantly. OFF-SGSP assumes model information, i.e., knowledge of transition dynamics. On the other hand, ON-SGSP operates in a model-free setting and hence, the analyses for both the timescales change. In particular, ON-SGSP uses a TD-critic and using standard stochastic approximation arguments (as in earlier literature), it is straightforward to prove the value updates in \eqref{eqn:don-sgsp-v} converge to the true value function. 
However, the analysis of the $\xi$-recursion changes significantly. The latter is a consequence of the fact that ON-SGSP operates in a model-free setting and hence does not have access to $f(\v_n, \pi_n)$ (and hence $\dfrac{\partial f(\v_n, \pi_n)}{\partial \pi^i(x, a^i)}$ which is required for the policy update).
Finally, the policy updates can be shown to track the system of ODEs \eqref{eqn:pi-ode} as in OFF-SGSP, after handling an additional martingale sequence that arises due to the $\xi$-recursion. The detailed proof is available in Section \ref{sect:onsgsp-proof}.  
\end{remark}

\section{Proof of Convergence}
\label{sect:convergence}
We provide a proof of convergence of the two proposed algorithms - OFF-SGSP and ON-SGSP, respectively. 
In addition to Assumption \ref{assumption:step-sizes}, we make the following assumption for the analysis of our algorithms:
\begin{assumption}
\label{assumption:markov-ir-pr-appendix}
The underlying Markov chain with transition probabilities $p(y | x, \pi)$, $x, y \in \S$, corresponding to the general-sum discounted stochastic game, is irreducible and positive recurrent for all possible strategies $\pi$.
\end{assumption}
The above assumption is standard in the analysis of multi-agent RL algorithms and can be seen in earlier works as well (for instance, see \cite{huwellman,littman}).

In the following section, we provide the detailed analysis for OFF-SGSP and later, in Section \ref{sect:onsgsp-proof}, provide the necessary modifications to the analysis for ON-SGSP.

\subsection{Proof of Theorem \ref{thm:off-sgsp-main} for OFF-SGSP}
\label{sect:offsgsp-proof}
As mentioned earlier, OFF-SGSP employs two time-scale stochastic approximation \cite[Chapter 6]{borkar2008stochastic}.
That is, it comprises of iteration sequences that are updated using two different time-scales or step-size schedules defined via $\{b(n)\}$ and $\{c(n)\}$, respectively. The step-sizes, satisfying Assumption \ref{assumption:step-sizes}, ensure the following:
\begin{enumerate}[\bfseries (i)]
\item The policy $\pi$ (on slower timescale) appears quasi-static for updates of $\v$; and 
\item The value $\v$ (on faster timescale) appears almost equilibrated for updates of $\pi$. We let $\v_\pi$ denote the value for a given policy $\pi$.
\end{enumerate}

Claim (i) above can be inferred as follows: First rewrite the $\pi$-recursion in \eqref{eqn:off-sgsp-pi} as
 $$
\pi^i_{n + 1}(x, a^i) = \Gamma \left(\pi^i_{n}(x, a^i) - c(n) H(n) \right), $$
 where $H(n) = \frac{b(n)}{c(n)} \sqrt{\pi^i_{n}(x, a^i)} \left|g^i_{x, a^i}(\v^i_n, \pi^{-i}_n)\right| \bsgn\left(\dfrac{\partial f(\v_n, \pi_n)}{\partial \pi^i}\right)$, with $g^i(\cdot,\cdot)$ as defined in \eqref{eq:g} and $f(\cdot,\cdot)$ as defined in \eqref{eqn:OP}. Since we consider a finite state-action spaced stochastic game, $g^i_{x, a^i}$ is bounded, while one can trivially upper bound $\pi$ and $\bsgn$. Thus, $\sup_n \left|H(n)\right|$ is finite. Since, $\frac{b(n)}{c(n)} = o(1)$ by assumption \ref{assumption:step-sizes}, it can be clearly seen that the $\pi$-recursion in \eqref{eqn:off-sgsp-pi} tracks the ODE 
 $$\dfrac{d \pi^i(x, a^i)}{dt} = 0.$$
 Claim (i) now follows. 

Inferring claim (ii) above is technically more involved, but follows using arguments similar to those used in Theorem 2 of Chapter 6 in \citep{borkar2008stochastic}.

In order to prove Theorem \ref{thm:off-sgsp-main}, we analyse each timescale separately in the following.
\subsection*{Step 1: Analysis of $\v$-recursion}
We first show that the updates of $\v$, that are on the faster time-scale, converge to a limit point of the following system of ODEs:$\forall x \in \S, i = 1, 2, \dots, N,$
\begin{align}
 \label{eqn:v-ode}
\dfrac{d v^i(x)}{dt} = r^i(x, \pi) + \beta \sum\limits_{y \in U(x)} p(y|x, \pi) v^i(y) - v^i(x), 
\end{align}
where $\pi$ is assumed to be time-invariant. We will also see that the system of ODEs above has a unique limit point, henceforth referred to as $\v_\pi$, which is stable. 

Let $R^i_\pi = \left < r^i(x, \pi), x \in \S \right >$ be a column vector of rewards to agent $i$ and $P_\pi = [ p(y|x, \pi), x \in \S, y \in \S ]$ be the transition probability matrix, both for a given $\pi$.
\begin{lemma}
\label{lemma:v_pi}
The system of ODEs (\ref{eqn:v-ode}) has a unique globally asymptotically stable limit point given by 
\begin{equation}
\label{eqn:v_pi}
\v^i_\pi = \left [ I - \beta P_\pi \right ]^{-1} R^i_\pi, i = 1, 2, \dots, N.
\end{equation}
\end{lemma}
\begin{proof}
The system of ODEs (\ref{eqn:v-ode}) can be re-written in vector form as given below.
\begin{equation}
\label{eqn:v-vector-ode}
\dfrac{d \v^i}{dt} = R^i_\pi + \beta P_\pi \v^i - \v^i.
\end{equation}
Rearranging terms, we get
\[\dfrac{d \v^i}{dt} = R^i_\pi + ( \beta P_\pi - I ) \v^i,\]
where $I$ is the identity matrix of suitable dimension. Note that for a fixed $\pi$, this ODE is linear in $\v^i$ with state transition matrix as $( \beta P_\pi - I )$. Since $P_\pi$ is a stochastic matrix, the magnitude of all its eigen-values is upper bounded by 1. Hence all the eigen-values of the state transition matrix $( \beta P_\pi - I )$ have negative real parts and the matrix $( \beta P_\pi - I )$ is in particular non-singular. Thus by standard linear systems theory, the above ODE has a unique globally asymptotically stable limit point which can be computed by setting $\dfrac{d \v^i}{dt} = 0, i = 1, 2, \dots, N$, i.e., \[R^i_\pi + ( \beta P_\pi - I ) \v^i = 0, i = 1, 2, \dots, N.\] The trajectories of the ODE (\ref{eqn:v-vector-ode}) converge to the above point starting from any initial condition in lieu of the above.
\end{proof}

For a given $\pi$, the updates of $\v$ in equation (\ref{eqn:off-sgsp-v}) (OFF-SGSP) can be seen as Euler discretization of the system of ODEs (\ref{eqn:v-ode}). We now show that $\v_n$ in equation (\ref{eqn:off-sgsp-v}) of OFF-SGSP converges to $\v_\pi$ as given in equation (\ref{eqn:v_pi}). While the following claim is identical for both OFF-SGSP/ON-SGSP, the proofs are quite different. In the former case, it amounts to proving value iteration converges (a standard result in dynamic programming), while the latter case amounts to proving a stochastic approximation variant of value iteration converges (also a standard result in RL).

\begin{proposition}
\label{lemma:off-sgsp-v-converge}
For a given $\pi$, i.e., with $\pi_n^i \equiv \pi^i$, updates of $\v$ governed by (\ref{eqn:off-sgsp-v}) (OFF-SGSP) satisfy $\v_n \rightarrow \v_\pi$, as $n \rightarrow \infty$, where $\v_\pi$ is the globally asymptotically stable equilibrium point of the system of ODEs (\ref{eqn:v-ode}).
\end{proposition}
\begin{proof}
We verify here assumptions (A1) and (A2) of \citet{borkar2000ode} in order to use their result \citep[Theorem 2.2]{borkar2000ode}. Let $h(\v^i) = R^i_\pi + ( \beta P_\pi - I ) \v^i$. Since $h(\v^i)$ is linear in $\v^i$, it is Lipschitz continuous. Let $h_r(\v^i) = \frac{h(r\v^i)}{r}$ for a scalar real number, $r > 0$. It is easy to see that $h_r(\v^i) = \frac{R^i_\pi}{r} + (\beta P_\pi - I) \v^i$. Now, $h_\infty(\v^i) = \lim\limits_{r \rightarrow \infty} h_r(\v^i) = (\beta P_\pi - I) \v^i$. Now since all eigenvalues of  $(\beta P_\pi - I)$ have negative real parts, the ODE $\dfrac{d \v^i}{d t} = h_\infty(\v^i)$ has the origin as its unique globally asymptotically stable equilibrium. Further, as shown in Lemma \ref{lemma:v_pi}, $\v_\pi^i$ is the unique globally asymptotically stable equilibrium for the ODE (\ref{eqn:v-vector-ode}).  Assumption (A1) of \citet{borkar2000ode} is thus satisfied. Since the updates of $\v$ in equation (\ref{eqn:off-sgsp-v}) do not have any noise term in them, assumption (A2) of \
citet{
borkar2000ode} is 
trivially satisfied. Thus by \citet[Theorem 2.2]{borkar2000ode}, $\v_n$ in equation (\ref{eqn:off-sgsp-v}) converges to the globally asymptotically stable limit point $\v_\pi$ given in equation (\ref{eqn:v_pi}).
\end{proof}

Thus, on the faster time-scale $\{c(n)\}$, the updates of $\v$ obtained from (\ref{eqn:off-sgsp-v}) converge to $\v_\pi$, as given by (\ref{eqn:v_pi}).


\subsection*{Step 2: Analysis of $\pi$-recursion}
Using the converged values of $\v$ corresponding to strategy update $\pi_n$, i.e., $\v_{\pi_n}$ on the slower time-scale, we show that updates of $\pi$ converge to a limit point of the following system of ODEs:\\ For all $ a^i \in \A^i(x), x \in \S, i = 1, 2, \dots, N,$
\begin{align}
\label{eqn:pi-ode}
\dfrac{d \pi^i(x, a^i)}{dt} = \bar\Gamma \left ( -\sqrt{\pi^i(x, a^i)} \left|g^i_{x, a^i}(\v_\pi^i, \pi^{-i})\right| \bsgn\left(\dfrac{\partial f(\v_\pi, \pi)}{\partial \pi^i}\right)\right ), 
\end{align}
where $g^i_{x,a^i}(\cdot,\cdot)$ is the Bellman error (see \eqref{eq:g}), $f(\cdot,\cdot)$ is the objective in \eqref{eqn:OP} and $\bar\Gamma$ is a projection operator that restricts the evolution of the above ODE to the simplex $\mathcal{D}$ and is defined as follows:
\begin{equation}
\label{eq:projected-ode}
\bar{\Gamma}(v(x)) = \lim_{\eta\rightarrow 0} \left(\frac{\Gamma(x+\eta v(x))-x}{\eta}\right),
\end{equation}
for any continuous $v:\mathcal{D}\rightarrow {\cal R}^N$.

Let $K$ denote the limit set of the ODE \eqref{eqn:pi-ode}.
Before we analyse the $\pi$-recursion in \eqref{eqn:off-sgsp-pi}, we show that the points in $K$ that are infeasible for the optimization problem \eqref{eqn:OP} are asymptotically unstable.
In other words, each stable limit point of the ODE \eqref{eqn:pi-ode} is an SG-SP point. 

Define the set of all feasible points of the optimization problem (\ref{eqn:OP}) as follows: 
\begin{align}
 G = \left\{\pi \in L \bigg | g^i_{x, a^i}(v_\pi^i, \pi^{-i}) \le 0, \forall a^i \in \A^i(x), x \in \S, i = 1, 2, \dots, N\right\}
\end{align}
The limit set $K$ of the ODE \eqref{eqn:pi-ode} can be partitioned using the feasible set $G$ as $K = K_1 \cup K_2$ where $K_1 = K \cap G$ and $K_2 = K \setminus K_1$. 
In the following lemma, we show that the set $K_2$ is the set of locally unstable equilibrium points of (\ref{eqn:pi-ode}).

\begin{lemma}
All $\pi^* \in K_2$ are unstable equilibrium points of the system of ODEs (\ref{eqn:pi-ode}).
\label{lemma:pi-convergence}
\end{lemma}
\begin{proof}
For any $\pi^* \in K_2$, there exists some $a^i \in \A^i(x), x \in \S, i \in \{ 1, 2, \dots, N \}$, such that $g^i_{x, a^i}(\v^i_\pi, \pi^{-i}) > 0$ and $\pi^i(x, a^i) = 0$ because $K_2$ is not in the feasible set $G$. Let $B_\delta(\pi^*) = \left \{ \pi \in L | \thinspace \|\pi - \pi^*\| < \delta \right \}$. Choose $\delta > 0$ such that $g^i_{x, a^i}(\v^i_\pi, \pi^{-i}) > 0$ for all $\pi \in B_\delta(\pi^*) \setminus K$ and consequently $\dfrac{\partial f(\v_\pi, \pi)}{\partial \pi^i} <0$. 

So, $\bar\Gamma\left(-\sqrt{\pi^i(x, a^i)} \left|g^i_{x, a^i}(\v^i_\pi, \pi^{-i})\right|\sgn\left(\dfrac{\partial f(\v_\pi, \pi)}{\partial \pi^i}\right)\right) > 0$ for any $\pi \in B_\delta(\pi^*) \setminus K$ which suggests that $\pi^i(x, a^i)$ will increase when moving away from $\pi^*$. Thus, $\pi^*$ is an unstable equilibrium point of the system of ODEs (\ref{eqn:pi-ode}).
\end{proof}

\begin{remark}\label{remark:sign}\textbf{\textit{(On the sign function)}}
Recall that $\bsgn$ was employed since the normal $sgn()$ function is discontinuous. 
Since $\bsgn$ can result in the value $0$, one can no longer conclude that $\sqrt{\pi^*} g =0$ for the points in the equilibrium set $K$. Note that the former condition (coupled with feasibility) implies it is an SG-SP point. A naive fix would be to change OFF-SGSP/ON-SGSP to repeat an action if $\bsgn(\cdot)$ returned $0$. This would ensure that there are no spurious points in the set $K$ due to $\bsgn$ being $0$. 
Henceforth, we shall assume that there are no such $\bsgn$ induced spurious limit points in the set $K$.
\end{remark}

\begin{lemma}
\label{lemma:complementary-pi-g}
For all $a^i \in \A^i(x), x \in \S$ and $i = 1, 2, \dots, N$, 
\begin{align}
\label{eqn:complementary-pi-g}
\pi \in K \Rightarrow \pi \in L \text{ and }\sqrt{\pi^i(x, a^i)} g^i_{x, a^i}(v_\pi^i, \pi^{-i}) = 0, 
\end{align}
where $L = \left \{ \pi | \pi(x)\text{ is a probability vector over }\A^i(x), \forall x \in \S \right \}.$
\end{lemma}

%
%
\begin{proof}
%
%
%
The operator $\bar\Gamma$, by definition, ensures that $\pi \in L$. Suppose for some $a^i~\in~\A^i(x), x~\in~\S$ and $i \in \{1, 2, \dots, N\}$, we have $\bar\Gamma(-\sqrt{\pi^i(x, a^i)} \left|g^i_{x, a^i}(v_\pi^i, \pi^{-i})\right|\bsgn\left(\dfrac{\partial f(\v_\pi, \pi)}{\partial \pi^i}\right)) = 0$, but \\
$\sqrt{\pi^i(x, a^i)} g^i_{x, a^i}(v_\pi^i, \pi^{-i})~\ne~0$. Then, $g^i_{x, a^i}(v_\pi^i, \pi^{-i}) \ne 0$ and since $\pi \in L$, $1 \ge \pi^i(x, a^i) > 0$. We analyze this condition by considering the following two cases.
\begin{description}
\item[Case $\bm{1 > \pi^i(x, a^i) > 0}$ and $\bm{g^i_{x, a^i}(v_\pi^i, \pi^{-i}) \ne 0}$.] 
\ \\
In this case, it is possible to find a $\Delta > 0$ such that for all $\delta \le \Delta$, 
\[1 > \pi^i(x, a^i) - \delta \sqrt{\pi^i(x, a^i)} \left|g^i_{x, a^i}(v_\pi^i, \pi^{-i})\right|\bsgn\left(\dfrac{\partial f(\v_\pi, \pi)}{\partial \pi^i}\right) > 0.\] 
This implies that 
\begin{align*}
  &\bar\Gamma\left (-\sqrt{\pi^i(x, a^i)} \left|g^i_{x, a^i}(v_\pi^i, \pi^{-i})\right|\bsgn\left(\dfrac{\partial f(\v_\pi, \pi)}{\partial \pi^i}\right) \right ) \\
  &= -\sqrt{\pi^i(x, a^i)} \left|g^i_{x, a^i}(v_\pi^i, \pi^{-i})\right|\bsgn\left(\dfrac{\partial f(\v_\pi, \pi)}{\partial \pi^i}\right) \ne 0, \\
  \Rightarrow& \sqrt{\pi^i(x, a^i)} g^i_{x, a^i}(v_\pi^i, \pi^{-i})\ne 0,
\end{align*}
which contradicts the initial supposition.
\item[Case $\bm{\pi^i(x, a^i) = 1}$ and $\bm{g^i_{x, a^i}(v_\pi^i, \pi^{-i}) \ne 0}$.] 
\ \\
Since $\v_\pi^i$ is solution of the system of ODEs (\ref{eqn:v-ode}), the following should hold: \[\sum\limits_{\hat{a}^i \in \A^i(x)} \pi^i(x, \hat{a}^i) g^i_{x, \hat{a}^i}(v_\pi^i, \pi^{-i}) = \pi^i(x, a^i) g^i_{x, a^i}(v_\pi^i, \pi^{-i}) = 0.\] This again leads to a contradiction.
\end{description}
The result follows.
\end{proof}

%

In order to prove Theorem \ref{thm:off-sgsp-main}, we require the well-known Kushner-Clark lemma (see \cite[pp. 191-196]{kushner-clark}). For the sake of completeness, we recall this result below.

\begin{theorem}\textbf{\textit{(Kushner-Clark lemma)}}
\label{thm:kc}
Consider the following recursion in $d$-dimensions:
\begin{equation}
\label{eq:kush-cla}
x_{n+1} = \Gamma(x_{n} + b(n)(h(x_n) + \zeta_n + \beta_n)), 
\end{equation}
where $\Gamma$ projects the iterate $x_n$ onto a compact and convex set, say $\C \in \R^d$. 
The ODE associated with (\ref{eq:kush-cla}) is given by
\begin{equation}
\label{eq:kushcla-ode}
\dot{x}(t) = \bar{\Gamma}(h(x(t))),
\end{equation}
where $\bar\Gamma$ is a projection operator that keeps the ODE evolution within the set $\C$ and is defined as in \eqref{eq:projected-ode}.
We make the following assumptions:
\begin{enumerate}[\bfseries (B1)]
\item $h$ is a continuous $\R^d$-valued function.
\item  The sequence $\beta_n,n\geq 0$ is a bounded random sequence with
$\beta_n \rightarrow 0$ almost surely as $n\rightarrow \infty$.
\item The step-sizes $b(n),n\geq 0$ satisfy
$  b(n)\rightarrow 0 \mbox{ as }n\rightarrow\infty \text{ and } \sum_n b(n)=\infty.$
\item $\{\zeta_n, n\ge 0\}$ is a sequence such that for any $\epsilon>0$,
\[ \lim_{n\rightarrow\infty} P\left( \sup_{m\geq n}  \left\|
\sum_{i=n}^{m} b(i) \zeta_i\right\| \geq \epsilon \right) = 0. \]
\end{enumerate} 
Suppose that the ODE \eqref{eq:kushcla-ode} has a compact set $K^*$ as its set of asymptotically stable equilibrium points.
Then, $x_n$ converges almost surely to $K^*$ as $n\rightarrow\infty$.
\end{theorem}

\subsection*{Proof of Theorem \ref{thm:off-sgsp-main}}
\begin{proof}
The updates of $\pi$ given by (\ref{eqn:off-sgsp-pi}) on the slower time-scale $\{b(n)\}$ can be rewritten as: For all $a^i \in \A^i(x), x \in \S$ and $i = 1, 2, \dots, N$,
\begin{align}
\label{eq:offsgsp-pi-equiv}
\pi^i_{n + 1}(x, a^i) = \Gamma \left(\pi^i_{n}(x, a^i) - b(n) (H(\pi^i_n)  +  \beta_n) \right),
\end{align}
where 
\begin{align*}
H(\pi^i_n) = &\sqrt{\pi^i_{n}(x, a^i)} \left|g^i_{x, a^i}(v_{\pi_n}^i, \pi^{-i}_n)\right|\bsgn\left(\dfrac{\partial f(v_{\pi_n}, \pi_n)}{\partial \pi^i}\right),\\
\beta_n = & \sqrt{\pi^i_{n}(x, a^i)} \left|g^i_{x, a^i}(v_{n}^i, \pi^{-i}_n)\right|\bsgn\left(\dfrac{\partial f(v_{n}, \pi_n)}{\partial \pi^i}\right) - H(\pi^i_n).
\end{align*}
We now verify the assumptions (B1)--(B4) for the recursion above:
\begin{itemize}
 \item $H(\pi^i_n)$ is continuous since each of its components $\sqrt{\pi^i}$, $\left|g^i_{x, a^i}(\cdot,\cdot)\right|$ and $\bsgn(\cdot)$ are continuous. In particular, the continuity of $g^i_{x, a^i}$ follows from the fact that both the value function $\v^i(\cdot)$ and Q-value function $Q^i_{\pi^{-i}}(x,a^i)$ are continuous in $\pi^i$. This verifies assumption (B1).
 \item $\beta_n \rightarrow 0$ almost surely since $|v_n - v_{\pi_n}| \rightarrow 0$ as $n \rightarrow \infty$, from Theorem \ref{lemma:off-sgsp-v-converge}. Further, $\beta_n$ is bounded as each of its components are bounded. In particular, $g^i_{x, a^i}$ is bounded as we consider finite state-action spaced stochastic games, while $\pi^i$ and $\bsgn$ are trivially upper-bounded. Thus (B2) is satisfied.
 \item Assumption \ref{assumption:step-sizes} implies (B3) is satisfied.
 \item $\zeta_n$ is absent, obviating assumption (B4).
\end{itemize}
The claim now follows from Kushner-Clark lemma.
\end{proof}

%
%
\begin{remark}
\label{remark:off-sgsp-conv}
 Note that from the foregoing, the set $K$ comprises of both stable and unstable attractors and in principle from Lemma \ref{lemma:pi-convergence}, the iterates $\pi^i_n$ governed by (\ref{eqn:pi-ode}) can converge to an unstable equilibrium. In most practical scenarios, however, a gradient descent scheme is observed to converge to a stable equilibrium. In fact, the $\delta$-offset policy computed using $perturb(\cdot,\delta)$ for every $Q > 0$ iterations (see Section \ref{sect:simulation} below) for both of our algorithms ensures numerically that as $n \rightarrow \infty$, $\pi_n \nrightarrow \pi^* \in K_2$. In other words, convergence of the strategy sequence $\pi_n$ governed by \eqref{eqn:off-sgsp-pi} is to the stable set $K_1$. 
\end{remark}


\subsection{Proof of Theorem \ref{thm:off-sgsp-main} for ON-SGSP}
\label{sect:onsgsp-proof}
As mentioned earlier, the analysis for ON-SGSP changes for both timescales and we outline the crucial differences below, before presenting the detailed analysis.
\begin{description}
\item[Step 1:]  This step establishes that the TD updates along faster timescale converge to the true value functions, using standard techniques from stochastic approximation. Unlike OFF-SGSP, this step also involves the analysis of the $\xi$-recursion. The latter is a consequence of the fact that we work in a model-free setting and hence do not have access to $f(\v_n, \pi_n)$ (and hence $\dfrac{\partial f(\v_n, \pi_n)}{\partial \pi^i(x, a^i)}$ which is required for the policy update).
\item[Step 2:] This step establishes that the policy updates track the same ODE (i.e., \eqref{eqn:pi-ode}) as that of OFF-SGSP and the analysis involves an additional martingale sequence that needs to be bounded.
\end{description}
 
\subsection*{Step 1: Analysis of $\v$ and $\xi$-recursions}
\begin{proposition}
\label{thm:onsgsp-v}
For a given $\pi$, i.e., with $\pi_n^i \equiv \pi^i$, updates of $\v$ governed by (\ref{eqn:don-sgsp-v}) (ON-SGSP) satisfy $\v_n \rightarrow \v_\pi$ almost surely as $n \rightarrow \infty$, where $\v_\pi$ is the globally asymptotically stable equilibrium point of the system of ODEs (\ref{eqn:v-ode}).
\end{proposition}
\begin{proof}
Fix a state $x \in \S$. Let $\{\bar{n}\}$ represent a sub-sequence of iterations in ON-SGSP when the state is $x \in \S$. Also, let $Q_n = \left \{ \bar n : \bar n < n \right \}$. For a given $\pi$, the updates of $\v$ on the faster time-scale $\{c(n)\}$ given in equation (\ref{eqn:don-sgsp-v})  can be re-written as
\[v^i_{\bar{n} + 1}(x) = v^i_{\bar{n}}(x) + c(\bar n) \left [ J(v^i_{\bar n}) + \tilde\chi_{\bar{n}} \right ],\]
where 
\begin{align*}
J(v^i_{\bar n})=& \sum\limits_{a^i \in \A^i(x)} \pi^i(x, a^i) g^i_{x, a^i}(\v^i_{\bar{n}}, \pi^{-i}), \text{ and }\\
\tilde\chi_{\bar{n}} =& \left(r^i_n \!+\! \beta v^i_n(y_n) \!- \!v^i_n(x_n)\right) - \sum\limits_{a^i \in \A^i(x)} \pi^i(x, a^i) g^i_{x, a^i}(\v^i_{\bar{n}}, \pi^{-i}).
\end{align*}
Using arguments as before, it is easy to see that $J(v^i_{\bar n})$ is continuous, $\tilde\chi_{\bar{n}}$ is such that $\E \tilde\chi_{\bar{n}}^2 < \hat C <  \infty$. 
Thus, 
\begin{align*}
 \lim_{\bar n\rightarrow\infty} P\left( \sup_{m\geq \bar n}  \left\|
\sum_{l=\bar n}^{m} c(l) \tilde\chi_{l}\right\| \geq \epsilon \right) \le \dfrac{\hat C}{\epsilon^2} \lim_{\bar n\rightarrow\infty} \sum_{l=\bar n}^{\infty}  c(l)^2 =0.
\end{align*}
For the last equality, we have used the fact that the step-sizes are square-summable (see assumption \ref{assumption:step-sizes}).
Thus. all the assumptions of Kushner-Clark Lemma (see Theorem \ref{thm:kc} above) are satisfied and we can conclude that $v_n$ governed by \eqref{eqn:don-sgsp-v} converges to the globally asymptotically stable limit point $\v_\pi$ (see equation (\ref{eqn:v_pi})) of the system of ODEs (\ref{eqn:v-ode}). 
\end{proof}
Before establishing the convergence of the gradient estimation recursion, i.e., \eqref{eqn:don-sgsp-q}, we require the following technical result.
\begin{lemma}
\label{lemma:grad-f}
\begin{align}
\dfrac{\partial f(\v, \pi)}{\partial \pi^i(x, a^i)} =&  - \sum\limits_{j=1}^N g^j_{x,a^i}(v^j,\pi^{-i}), \text{ where }
\label{eq:grad-f}\\
g^j_{x, a^i}(\v^j, \pi^{-i})=&\bar{r}^j(x,\pi^{-i}, a^i) + \beta \sum \limits_{y \in U(x)} \bar{p}^j(y|x,\pi^{-i}, a^i) v^j(y) -v^j(x).\label{eq:gj}
\end{align}
\end{lemma}
\begin{proof}
Let $a\in\A(x)$ denote the aggregate action vector.
Then, we have
\begin{align*}
\dfrac{\partial f(\v, \pi)}{\partial \pi^i(x, a^i)} =&  - g^i_{x,a^i}(v^i,\pi^{-i})- \sum\limits_{j\ne i} \sum\limits_{a^j\in \A^j(x)}
\pi^j(x,a^j)
\bigg( \sum\limits_{k\ne i,j} \sum\limits_{a^k\in \A^k(x)} \prod\limits_{k\ne i,j} \pi^k(x,a^k)\big( r^j(x,a)  \\
&\hspace{12em}+\beta \sum\limits_{y \in U(x)} p(y | x, a) v^j (y) - v^j(x) \big)\bigg)\\
& = - g^i_{x,a^i}(v^i,\pi^{-i})- \sum\limits_{j\ne i}\bigg(\sum\limits_{k\ne i} \sum\limits_{a^k\in \A^k(x)}
  \prod\limits_{k\ne i} \pi^k(x,a^k)\big( r^j(x,a) \\
&\hspace{12em}  + \beta \sum\limits_{y \in U(x)} p(y | x, a) v^j (y) - v^j(x) \big)\bigg).
\end{align*}
Now, from the definition of $g^j_{x,a^i}(v^i,\pi^{-i})$ in \eqref{eq:gj}, it is easy to see that 
\begin{align*}
\dfrac{\partial f(\v, \pi)}{\partial \pi^i(x, a^i)} =  - \sum\limits_{j=1}^N g^j_{x,a^i}(v^j,r^j,\pi^{-i}).
\end{align*}
\end{proof}

Recall that the gradient estimation recursion is as follows:
\begin{align*}
\xi^i_{n+1}(x_n, a^i_n) = \xi^i_{n}(x_n, a^i_n) + c(n) \bigg( \sum\limits_{j=1}^N \big( r^j_n + \beta v^j_n(y_n) - v^j_n(x_n) \big) -  \xi^i_n(x_n, a^i_n) \bigg).
\end{align*}
The following theorem establishes that $\xi^i_{n }(x, a^i)$ converges to $-\dfrac{\partial f(\v_n, \pi_n)}{\partial \pi^i(x, a^i)}$ in the long run.
\begin{proposition}
\label{thm:onsgsp-q}
$\left\|\xi^i_n(x, a^i) - \left(- \dfrac{\partial f(\v_n, \pi_n)}{\partial \pi^i(x, a^i)}\right) \right\| \rightarrow 0$ as $n \rightarrow \infty$ almost surely.
\end{proposition}
\begin{proof}
As in the proof of Proposition \ref{thm:onsgsp-v}, let $\{\bar{n}\}$ represent a sub-sequence of iterations in ON-SGSP when the state is $x \in \S$ and $Q_n := \left \{ \bar n : \bar n < n \right \}$. For a given $\pi$, the updates of $\xi$ on the faster time-scale $\{c(n)\}$ given in equation (\ref{eqn:don-sgsp-q}) can be re-written as
\begin{align*}
\xi^i_{\bar n + 1}(x, a^i) = (1-c(\bar n)) \xi^i_{\bar n}(x, a^i) + c(\bar n) \bigg( \sum\limits_{j=1}^N g^j_{x,a^i}(v^j_{\bar n},r^j_{\bar n},\pi^{-i}) + \xi_{\bar n} \bigg), 
\end{align*}
where $\xi_{\bar n} := \sum\limits_{j=1}^N \big( r^j_{\bar n} + \beta v^j_{\bar n}(y) - v^j_{\bar n}(x) \big) -  \sum\limits_{j=1}^N g^j_{x,a^i}(v^j_{\bar n},r^j_{\bar n},\pi^{-i})$. Let $\mathcal{F}_l := \sigma(\v_k, \xi_k, k \le l), l \ge 0$ denote an increasing family of $\sigma$-fields. By definition of $g^j_{x,a^i}$, we have 
$$E\left[r^j_{\bar n} + \beta v^j_{\bar n}(y) - v^j_{\bar n}(x)\mid \mathcal{F}_{\bar n}, \pi^i_{\bar n}(x_{\bar n},a^i_{\bar n}) \right] = \sum\limits_{j=1}^N g^j_{x,a^i}(v^j_{\bar n},r^j_{\bar n},\pi^{-i}).$$ 
Hence, $\{\xi_{\bar n}\}$ is a martingale difference sequence. 

As in Proposition \ref{thm:onsgsp-v}, define 
$\tilde M_{n} = \sum\limits_{m \in Q_n} c(m) \tilde\chi_m$. It can be easily verified that $(\tilde M_m, \mathcal{F}_m), m \ge 0$ is a square-integrable martingale sequence obtained from the corresponding martingale difference $\{\xi_m\}$. Further, from the square summability of $c(n), n \ge 0$, and assumption \ref{assumption:markov-ir-pr-appendix} which ensures that the underlying Markov chain is ergodic for any given $\pi$, it can be verified from the martingale convergence theorem that $\{\tilde M_m, m \ge 0\}$, converges almost surely.
Hence,
$|\xi_m| \rightarrow 0$ almost surely on the `natural timescale', as $m \rightarrow \infty$.  The `natural timescale' is clearly faster than the algorithm's timescale and hence $\xi_m$ can be ignored in the analysis of Bellman error recursion (\eqref{eqn:don-sgsp-q} in the main paper), see \cite[Chapter 6.2]{borkar2008stochastic} for detailed treatment of natural timescale algorithms.
The final claim follows from Kushner-Clark lemma.
\end{proof}

\subsection*{Step 2: Analysis of $\pi$-recursion}
\begin{proof}
We first re-write the update of $\pi$ as follows: For all $i = 1, 2, \dots, N$, 
\begin{align*}
\pi^i_{\bar n + 1}(x, a^i) =&  \Gamma \left(\pi^i_{\bar n}(x, a^i) - b(\bar n)\left( \sqrt{\pi^i_{\bar n}(x, a^i)} \left|g^i_{x, a^i}(v_{\pi_{\bar n}}^i, \pi^{-i}_{\bar n})\right|\bsgn\left(\dfrac{\partial f(v_{\pi_{\bar n}}, \pi_{\bar n})}{\partial \pi^i}\right) + \zeta_{\bar n}\right)\right),
\end{align*}
where 
$\zeta_{\bar n} = \sqrt{\pi^i_{\bar n}(x, a^i)} \left [ \left|\hat g^i_{x, a^i}\right| - \left|g^i_{x, a^i}(v_{\pi_{\bar n}}^i, \pi^{-i}_{\bar n})\right| \right]\bsgn\left(\dfrac{\partial f(v_{\pi_{\bar n}}, \pi_{\bar n})}{\partial \pi^i}\right)$.
In the above, we have used the converged values of the value update $v_n$ and gradient estimate $\xi_n$ and this is allowed due to timescale separation and the fact that $v_n \rightarrow \v_{\pi_n}$ (see Proposition \ref{thm:onsgsp-v}) and $\xi_{n} \rightarrow - \dfrac{\partial f(\v_n, \pi_n)}{\partial \pi^i(x, a^i)}$ (see Proposition \ref{thm:onsgsp-q}).
Now, in order to apply Kushner-Clark lemma (see Theorem \ref{thm:kc} above), it is enough if we verify that $\E \zeta_{\bar n}^2 < \infty$, since the rest of the terms are as in OFF-SGSP (which imply assumptions (B1) to (B3) in Theorem \ref{thm:kc} are verified). Now, arguing as before, it is straightforward to infer that $\E \zeta_{\bar n}^2 < \infty$, since we consider finite state-action spaces and the square of each of the quantities in the first term in $\zeta_{\bar n}$ can be upper-bounded.  
Thus, assumption (B4) in Theorem \ref{thm:kc} is verified and updates of $\pi$ in ON-SGSP converge to a stable limit point of the system of ODEs (\ref{eqn:pi-ode}).
\end{proof}


\section{Simulation Experiments}
\label{sect:simulation}
We test ON-SGSP, NashQ \citep{huwellman2003} and FFQ \citep{littman} algorithms on two general-sum game setups. We implemented Friend Q-learning variant of FFQ, as each iteration of its Foe Q-learning variant involves a computationally intensive operation to solve a linear program. 

\subsection{Single State (Non-Generic) Game}
This is a simple two-player game adopted from \cite{hart2005stochastic}, where the payoffs to the individual agents are given in Table \ref{tab:hart}. In this game, a strategy that picks $a_3$ (denoted by $(0,0,1)$) constitutes a pure-strategy NE, while a strategy that picks either $a_1$ or $a_2$ with equal probability (denoted by $(0.5,0.5,0)$) is a mixed-strategy NE. 

We conduct a stochastic game experiment where  at each stage, the payoffs to the agents are according to Table \ref{tab:hart} and the payoffs accumulate with a discount factor $\beta = 0.8$. We performed $100$ experimental runs, with each run corresponding to a length of $10000$ stages. The aggregated results from this experiment are presented in  Fig. \ref{tab:results}. It is evident that NashQ oscillates and does not converge to NE in most of the runs, while Friend Q-learning converges to a non-Nash strategy tuple in most of the runs. On the other hand, ON-SGSP converges to NE in all the iterations.  

\begin{figure*}
\centering
\begin{tabular}{c}
\begin{subfigure}[b]{0.4\textwidth}
\centering
\begin{tabular}{c|c|c|c}
\textbf{Player }$\bm{2 \rightarrow}$ & \multirow{2}{*}{$\bm{a_1}$} & \multirow{2}{*}{$\bm{a_2}$} & \multirow{2}{*}{$\bm{a_3}$}\\ 
\textbf{Player }$\bm{1}$ &&&\\
$\bm{\downarrow}$ &&&\\\hline
\multirow{2}{*}{$\bm{a_1}$} & \multirow{2}{*}{$1,0$} & \multirow{2}{*}{$0,1$} & \multirow{2}{*}{$1,0$}\\ 
 &&&\\\hline
\multirow{2}{*}{$\bm{a_2}$} & \multirow{2}{*}{$0,1$} & \multirow{2}{*}{$1,0$} & \multirow{2}{*}{$1,0$}\\ 
 &&&\\\hline
\multirow{2}{*}{$\bm{a_3}$} & \multirow{2}{*}{$0,1$} & \multirow{2}{*}{$0,1$} & \multirow{2}{*}{$1,1$}\\ 
 &&&\\
\end{tabular}
\caption{Payoff matrix.}
\label{tab:hart}
\end{subfigure}
\\[3ex]
\begin{subfigure}[b]{0.6\textwidth}
\begin{tabular}{c|c|c|c}
 & \textbf{NashQ} & \textbf{FFQ (Friend Q)} & \textbf{ON-SGSP} \\ \hline
\textbf{Oscillate or 
converge} & \multirow{2}{*}{\color{red}$\bm{95\%}$} & \multirow{2}{*}{\color{red}$\bm{40\%}$} & \multirow{2}{*}{$0\%$} \\
\textbf{to non-Nash strategy}&&\\ \hline
\multirow{2}{*}{\textbf{Converge to
$(0.5,0.5,0)$}} & \multirow{2}{*}{$2\%$} & \multirow{2}{*}{$0\%$} & \multirow{2}{*}{{\color{darkgreen}$\bm{99\%}$}} \\ 
&&\\\hline
\multirow{2}{*}{\textbf{Converge to
$(0,0,1)$}} & \multirow{2}{*}{$3\%$} & \multirow{2}{*}{$60\%$} & \multirow{2}{*}{$1\%$} \\
&&\\
\end{tabular}
\caption{Results from 100 simulation runs.}
\label{tab:results}
\end{subfigure}
\end{tabular}
\caption{Payoff matrix and simulation results for a single state non-generic two-player game}
\label{tab:hart-results}
\end{figure*}

\begin{figure}
\centering
\begin{tikzpicture}
\draw[xshift=-0.75cm,yshift=-0.75cm, step=1.5] (0,0) grid (4.5,4.5);

\draw (3,0) node[fill=red,rectangle,inner sep=5pt] {};
\draw (0,1.5) node[fill=darkgreen,rectangle,inner sep=5pt] {};

\draw[->,thick] (3, 0) -- +(135:0.75000);
\draw[->,thick] (3, 0) -- +(90:0.70);
\draw[->,thick] (3, 0) -- +(180:0.650);
\draw[->,thick] (0, 1.5) -- +(315:0.70000);
\draw[->,thick] (0, 1.5) -- +(45:0.700000);
\draw[->,thick] (0, 1.5) -- +(0:0.650000);
\draw[->,thick] (0, 1.5) -- +(90:0.6500000);
\draw[->,thick] (0, 1.5) -- +(270:0.6500000);
\end{tikzpicture}
 \caption{Stick-Together Game for $M=3$}
 \label{fig:stg}
\end{figure}
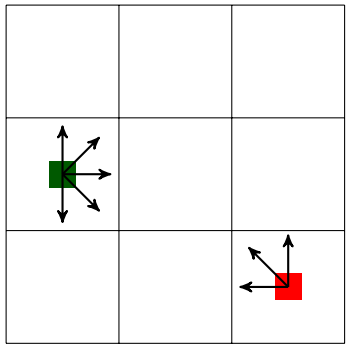

\subsection{Stick-Together Game (STG)}
We also define a simple general-sum discounted stochastic game, named ``Stick-Together Game'' or in short STG, where two participating agents located on a rectangular terrain would like to come together and stay close to each other (see Fig. \ref{fig:stg}). 
A precise description of the various components of STG is provided below:

\begin{enumerate}
\item \textbf{State Space $\S$:} The state specifies the location of both the agents on a rectangular grid of size $M\times M$. More precisely, let $O = \left \{ (x, y)| x, y \in \Z \right \}$. Denote the possible positions of an agent by $W := \left \{ s = (x, y) \in O | 0 \le x, y < M \right \}$. Then the state space is given by the Cartesian product $\S = W \times W$.
\item \textbf{Action Space $\A$:} The actions available to each agent are to either move to one of the neighboring cells or stay in the current location.  For $s \in W$, let $\|s\|_1 = |x| + |y|$ be its $L^1$ norm. Then, $A(s) = \left \{ a \in O | \|s + a\| \le 1 \right \}$ represents the actions available for an agent to move to one-step neighbouring positions of $s \in W$. The action space is then defined by $\A = \mathop\cup\limits_{s^1, s^2 \in W} A(s^1) \times A(s^2)$. Let $U(s) = \{ s' \in W | \|s' - s\|_1 \le 1 \}$ represent the set of all next states for an agent in state $s \in W$.
\item \textbf{Transition probability $p$:} We assume that state transitions of individual agents are independent. Let $q(s'|s, a^i)$ represent, for agent $i$, the probability of transition from state $s \in W$ to $s' \in W$ upon taking action $a^i \in \A^i(s)$. We define \[q(s'|s, a) = \dfrac{2^{-\|s' - a\|_1}}{\sum\limits_{s'' \in U(s)} 2^{-\|s'' - a\|_1}}.\] Then, the transition probability is given by 
$$p((s'^1, s'^2)|(s^1, s^2), (a^1, a^2)) = q(s'^1|s^1, a^1) q(s'^2|s^2, a^2).$$ 
This transition probability function has the highest value towards that next state to which the action points to.
\item \textbf{Reward $r$:} The reward for the two agents is defined as \[r^i(s^i, a^i) = 1 - e^{\|s^1 - s^2\|_1},\] for state $(s^1, s^2) \in \S$ and action $(a^1, a^2) \in A(s^1) \times A(s^2)$. Thus, the reward is zero if the distance between the two agents is zero. Otherwise, it is a negative and monotonically decreasing function with respect to the distance between the two agents.
\end{enumerate}

\paragraph{Results.}
 We first show simulation results for a small sized version of the STG game, where $M = 3$. The number of states with $M = 3$ is $|\S| = 81$. We use $\beta = 0.8$ for all our experiments. Also, we use the following step-size sequences in our experiments:
\[
\begin{array}{l}
b(n) = \left \{ \begin{array}{ll} 0.2 & \text{ for $n < 1000$,} \\ \dfrac{1}{n^{0.75}} & \text{ otherwise,} \end{array} \right .\\[2ex]
c(n) = \left \{ \begin{array}{ll} 0.1 & \text{ for $n < 1000$,} \\ \dfrac{1}{n} & \text{ otherwise.} \end{array} \right .
\end{array}
\]
It is easy to see that $\{c(n)\}$ corresponds to slower time-scale than  $\{b(n)\}$. It was observed in our experiments that, using constant step-sizes upto $n = 1000$ leads to better initial exploration and faster convergence. 

To ensure sufficient exploration of the state space (see assumption \ref{assumption:markov-ir-pr-appendix}) and also to push the policy $\pi$ out of the domain of attraction of any local equilibrium, we perturb the policy as follows: 
For every $Q > 0$ iterations, $perturb(\cdot, \delta)$ is used to derive a $\delta$-offset policy for picking actions, i.e., $\hat{\pi}^i(x)$ is used instead of $\pi^i(x)$, where
\begin{equation}
\label{eq:pi-perturb}
\hat{\pi}^i(x, a^i) = \dfrac{\pi^i(x, a^i) + \delta}{\sum \limits_{a^i \in \A^i(x)} \left (\pi^i(x, a^i) + \delta \right )}, a^i \in \A^i(x).
\end{equation}

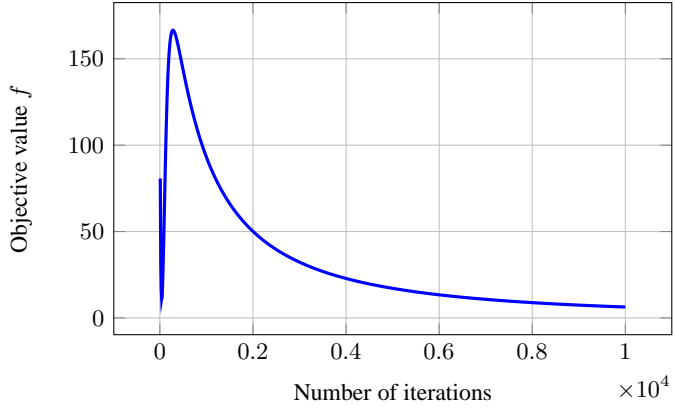
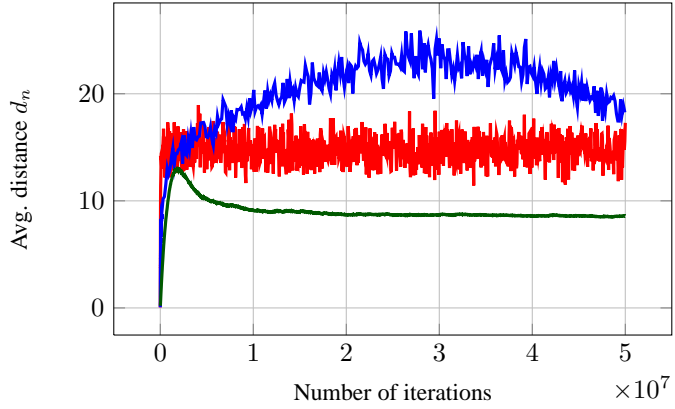
\begin{figure}[t]
\centering
\begin{tabular}{c}
\begin{subfigure}[b]{0.6\textwidth}
\begin{tikzpicture}
\begin{axis}[xlabel={{\small Number of iterations}}, ylabel={{\small Objective value $f$}},width=9cm,height=6cm,grid,no markers,tick scale binop={\times},font={\small}]
\addplot[color=blue,very thick] file {plots/off-sgsp-stg-3x3-subsampled.txt};
\end{axis}
\end{tikzpicture}
\caption{OFF-SGSP for STG with $M = 3$}
\label{fig:stg-off-m-3}
\end{subfigure}
 
\\
\begin{subfigure}[b]{0.6\textwidth}
\begin{tikzpicture}
\begin{axis}[xlabel={{\small Number of iterations}}, ylabel={{\small Avg. distance $d_n$}},width=9cm,height=6cm,grid,no markers,tick scale binop={\times},
legend style={
at={(0,0)},
anchor=north,at={(axis description cs:0.5,-0.4)},legend columns=-1},restrict x to domain=0:50000000]
\addplot[color=red,very thick] file {plots/ffq-stg-30x30-subsampled.txt};
\addlegendentry{FFQ}
\addplot[color=blue,very thick] file {plots/nashq-stg-30x30-subsampled.txt};
\addlegendentry{NashQ}
\addplot[color=darkgreen,very thick] file {plots/on-sgsp-stg-30x30-subsampled.txt};
\addlegendentry{ON-SGSP}
\end{axis}
\end{tikzpicture}
\caption{ON-SGSP for STG with $M = 30$}
\label{fig:stg-on-m-30}
\end{subfigure}

\end{tabular}
\caption{Performance of our algorithms for STG}
\end{figure}

 Fig. \ref{fig:stg-off-m-3} shows the evolution of the objective function $f$ as a function of the number of iterations for OFF-SGSP. Note that $f$ should go to zero for a Nash equilibrium point.

Fig. \ref{fig:stg-on-m-30}  shows the evolution of the distance $d_n$ (in $\ell_1$ norm) between the agents for a STG game where $M = 30$, which corresponds to $|\S| = 810,000$.  Notice that the results are shown only for the model-free algorithms: ON-SGSP, NashQ and FFQ. This is because OFF-SGSP and even the homotopy methods \citep{herings2004stationary} have exponential blow up with $M$ in their computational complexity and hence, are practically infeasible for STG with $M=30$. 

From Fig. \ref{fig:stg-on-m-30}, it is evident that following the ON-SGSP strategy, the agents converge to a $4\times4$-grid within the $30\times30$-grid. For achieving this result, ON-SGSP takes about $2 \times 10^7$ iterations, implying an average $2 \times 10^7/|\S| \approx 21$ iterations per state. 
However, NashQ gets the agents to an $8\times8$-grid after a large number of iterations ($\approx 5\times10^7$). Moreover, from Fig. \ref{fig:stg-on-m-30} it is clear that NashQ has not stabilized its strategy in the end.
Friend Q-learning gets the agents to $8\times8$-grid, by driving them to one of the corners of the $30\times30$-grid. While it takes a short number of iterations ($\approx 30000$) to achieve this, FFQ does not explore the state space well and hence, FFQ's strategy corresponding to the rest of the grid (excluding the corner to which it takes the agents) is not Nash. 

\begin{remark}
(\textbf{Runtime performance.})
We observed that to complete $5 \times 10^7$ iterations, ON-SGSP took $\approx 42$ minutes, while NashQ \citep{huwellman2003} took nearly $50$ hours, as it involves solving for Nash equilibria of a bimatrix game in each iteration. The Friend Q-learning variant of FFQ \citep{littman} took $\approx 33$ minutes. The Foe Q-learning variant of FFQ was not implemented owing to its high per-iteration complexity.
\end{remark}

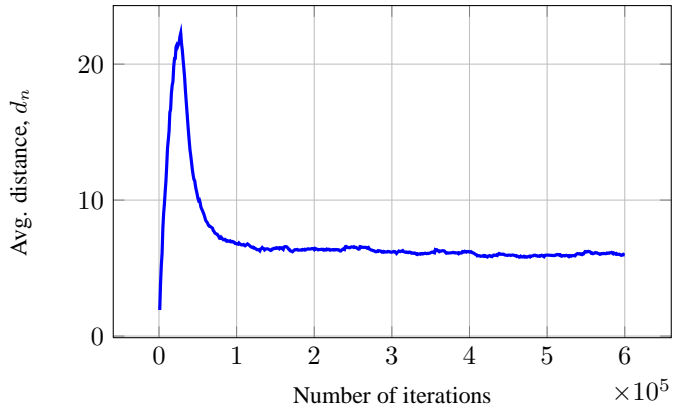
\begin{figure}
\centering
\begin{tikzpicture}
\begin{axis}[xlabel={{\small Number of iterations}}, ylabel={{\small Avg. distance, $d_n$}},width=9cm,height=6cm,grid,no markers,tick scale binop={\times}]
\addplot[color=blue,very thick] file {plots/on-sgsp-sub-stg-30x30-subsampled.txt};
\end{axis}
\end{tikzpicture} 
\caption{ON-SGSP with partial information for STG with $M = 30$}
\label{fig:sub-stg-on-m-30}
\end{figure}

\subsection*{Simple Function Approximation for STG}
While OFF-SGSP assumes full information of the game, ON-SGSP assumes that neither rewards nor state transition probabilities are known. 
Here, we explore an intermediate information case albeit restricted to STG where a partial structure of rewards is made known.  In particular, we assume that the reward depends on the difference $\Delta = (|x^1_1 - x^2_1|, |x^1_2 - x^2_2|) \in \S$ in positions $x_1 = (x^1_1, x^1_2), x_2 = (x^2_1, x^2_2) \in \S$ of the two agents. We approximate the value function $\v$ and strategy $\pi$ as follows:  $v^i(x) \approx \hat v^i(\Delta)$ and $\pi^i(x) \approx \hat\pi^i(\Delta), \forall x \in \S$. Thus, the algorithms need to compute $\hat\v$ and $\hat\pi$ on a low-dimensional subspace $W$ of $\S$. 
Fig. \ref{fig:sub-stg-on-m-30} presents results of ON-SGSP in this setting for $M = 30$, which corresponds to $|W| = 900$. 
The solution is seen to have converged by $200,000$ iterations ($\approx 5$ seconds runtime)  which suggests that it took on an average $\frac{200,000}{|W|} \approx 22$ iterations per $\Delta \in W$ to converge.


\section{Conclusions}
\label{sect:conclusions}
In this paper,
we derived necessary and sufficient SG-SP conditions to solve a generalized optimization problem and established their equivalence with Nash strategies. We derived a descent (not necessarily steepest) direction that avoids local minima. 
Incorporating this direction, we proposed two algorithms - offline, model-based algorithm OFF-SGSP and online, model-free algorithm ON-SGSP.  Both algorithms were shown to converge, in self-play, to the equilibria of a certain ordinary differential equation (ODE), whose stable limit points coincide with stationary Nash equilibria  of the underlying general-sum stochastic game. Synthetic experiments on two general-sum game setups show that ON-SGSP outperforms two well-known multi-agent RL algorithms. The experimental evaluation also suggests that convergence is relatively quick.

There are several future directions to be explored and we outline a few of them below:
\begin{enumerate}
 \item  In simulations, we observed that ON-SGSP can successfully run for large state spaces ($|\S| \approx 8,00,000$). However, in many cases, the state spaces can be huge and it would be necessary to look for function approximation techniques for both the value function $\v$ as well as strategy-tuple $\pi$. Function approximation techniques are popular in reinforcement learning approaches for high-dimensional MDPs and they bring in the following advantages: \begin{inparaenum} \item They can cater to huge state and action spaces; and \item They also aid a designer to bring in his understanding about the underlying system in terms of features used for function approximations. \end{inparaenum} 
 
 \item  Extensions to the case of constrained games: By constrained stochastic games, we mean those stochastic games that have additional constraints on value functions or strategy-tuple which might arise from the modelling of a practical scenario. There are some results and applications in this by \cite{altman2005zero} and \cite{altaian2007constrained} which provide the necessary motivation for extending our results to general-sum constrained stochastic games.
 
 \item Detailed experimental evaluation on a sophisticated benchmark for $N$-player general-sum stochastic games.
\end{enumerate}

\bibliography{sg}
\bibliographystyle{plainnat}

\end{document}